\documentclass[journal]{IEEEtran}

\usepackage{xcolor,soul,framed} 

\colorlet{shadecolor}{yellow}
\usepackage[pdftex]{graphicx}
\usepackage{microtype}
\graphicspath{{../pdf/}{../jpeg/}}
\DeclareGraphicsExtensions{.pdf,.jpeg,.png}
\usepackage{hyperref}

\usepackage[cmex10]{amsmath}
\usepackage{array}
\usepackage{float}
\usepackage{mdwmath}
\usepackage{mdwtab}
\usepackage{eqparbox}
\usepackage{booktabs}
\usepackage{multirow}
\usepackage{url}
\usepackage{braket}
\usepackage{amssymb}
\usepackage[utf8x]{inputenc}
\usepackage[justification=centering]{caption}
\usepackage{ulem}				
\usepackage{quantikz}  	
\usepackage{tikz}
\usepackage{adjustbox}
\usepackage{subcaption}
\usepackage{algorithm} 
\usepackage{algpseudocode}
\usepackage{bbm}
\usepackage[margin=0.7in]{geometry}
\newtheorem{theorem}{Theorem}

\usepackage{listings}
\lstdefinestyle{lststyle}{
  commentstyle=\color{green},
  keywordstyle=\color{magenta},
  numberstyle=\tiny\color{gray},
  stringstyle=\color{purple},
  basicstyle=\ttfamily\footnotesize,
  breakatwhitespace=false,
  breaklines=true,
  captionpos=b,
  frame=lines,
  keepspaces=true,
  numbers=left,
  numbersep=5pt,
  showspaces=false,
  showstringspaces=false,
  showtabs=false,
  tabsize=2
}
\usepackage{listings, xcolor}

\definecolor{verylightgray}{rgb}{.97,.97,.97}

\lstdefinelanguage{Solidity}{
    keywords=[1]{anonymous, assembly, assert, balance, break, call, callcode, case, catch, class, constant, continue, constructor, contract, debugger, default, delegatecall, delete, do, else, emit, event, experimental, export, external, false, finally, for, function, gas, if, implements, import, in, indexed, instanceof, interface, internal, is, length, library, log0, log1, log2, log3, log4, memory, modifier, new, payable, pragma, private, protected, public, pure, push, require, return, returns, revert, selfdestruct, send, solidity, storage, struct, suicide, super, switch, then, this, throw, transfer, true, try, typeof, using, value, view, while, with, addmod, ecrecover, keccak256, mulmod, ripemd160, sha256, sha3}, 
    keywordstyle=[1]\color{blue}\bfseries,
    keywords=[2]{address, bool, byte, bytes, bytes1, bytes2, bytes3, bytes4, bytes5, bytes6, bytes7, bytes8, bytes9, bytes10, bytes11, bytes12, bytes13, bytes14, bytes15, bytes16, bytes17, bytes18, bytes19, bytes20, bytes21, bytes22, bytes23, bytes24, bytes25, bytes26, bytes27, bytes28, bytes29, bytes30, bytes31, bytes32, enum, int, int8, int16, int24, int32, int40, int48, int56, int64, int72, int80, int88, int96, int104, int112, int120, int128, int136, int144, int152, int160, int168, int176, int184, int192, int200, int208, int216, int224, int232, int240, int248, int256, mapping, string, uint, uint8, uint16, uint24, uint32, uint40, uint48, uint56, uint64, uint72, uint80, uint88, uint96, uint104, uint112, uint120, uint128, uint136, uint144, uint152, uint160, uint168, uint176, uint184, uint192, uint200, uint208, uint216, uint224, uint232, uint240, uint248, uint256, var, void, ether, finney, szabo, wei, days, hours, minutes, seconds, weeks, years}, 
    keywordstyle=[2]\color{teal}\bfseries,
    keywords=[3]{block, blockhash, coinbase, difficulty, gaslimit, number, timestamp, msg, data, gas, sender, sig, value, now, tx, gasprice, origin},   
    keywordstyle=[3]\color{violet}\bfseries,
    identifierstyle=\color{black},
    sensitive=false,
    comment=[l]{//},
    morecomment=[s]{/*}{*/},
    commentstyle=\color{gray}\ttfamily,
    stringstyle=\color{red}\ttfamily,
    morestring=[b]',
    morestring=[b]"
}

\begin{document}
\bstctlcite{IEEEexample:BSTcontrol}
    \title{Efficient Floating Point Arithmetic\\ for Quantum Computers\\}
    \author{
    \IEEEauthorblockN{
    Raphael Seidel\IEEEauthorrefmark{2},
    Nikolay Tcholtchev\IEEEauthorrefmark{2},
    Sebastian Bock\IEEEauthorrefmark{2},
    Colin Kai-Uwe Becker\IEEEauthorrefmark{2},
    Manfred Hauswirth\IEEEauthorrefmark{2}\IEEEauthorrefmark{1}\medskip }
    \IEEEauthorblockA{\IEEEauthorrefmark{2}Fraunhofer Institute for Open Communication Systems (FOKUS) \\ Berlin, Germany \\ \{firstname.lastname\}@fokus.fraunhofer.de}\\ \medskip
    \IEEEauthorblockA{\IEEEauthorrefmark{1}Technische Universität Berlin
    \vspace{-1cm}}
}





\maketitle

\definecolor{refblue}{RGB}{102, 102, 153}

\let\oldref\ref
\renewcommand{\ref}[1]{{\color{refblue}(\oldref{#1})}}
\let\oldcite\cite
\renewcommand{\cite}[1]{{\color{refblue}\oldcite{#1}}}



\begin{abstract}
\thispagestyle{empty}
One of the major promises of quantum computing is the realization of SIMD (single instruction - multiple data) operations using the phenomenon of superposition. Since the dimension of the state space grows exponentially with the number of qubits, we can easily reach situations where we pay less than a single quantum gate per data point for data-processing instructions which would be rather expensive in classical computing. Formulating such instructions in terms of quantum gates, however, still remains a challenging task. Laying out the foundational functions for more advanced data-processing is therefore a subject of paramount importance for advancing the realm of quantum computing. In this paper, we introduce the formalism of encoding so called-semi-boolean polynomials. As it turns out, arithmetic $\mathbb{Z}/2^n\mathbb{Z}$ ring operations can be formulated as semi-boolean polynomial evaluations, which allows convenient generation of unsigned integer arithmetic quantum circuits. For arithmetic evaluations, the resulting algorithm has been known as Fourier-arithmetic. We extend this type of algorithm with additional features, such as ancilla-free in-place multiplication and integer coefficient polynomial evaluation. Furthermore, we introduce a tailor-made method for encoding signed integers succeeded by an encoding for arbitrary floating-point numbers. This representation of floating-point numbers and their processing can be applied to any quantum algorithm that performs unsigned modular integer arithmetic. We discuss some further performance enhancements of the semi boolean polynomial encoder and finally supply a complexity estimation. The application of our methods to a 32-bit unsigned integer multiplication demonstrated a 90\% circuit depth reduction compared to carry-ripple approaches.
\end{abstract}

\begin{IEEEkeywords}
quantum computing, quantum arithmetic, floating point arithmetic
\end{IEEEkeywords}

%
\IEEEpeerreviewmaketitle



\section{Introduction}
\label{sec:introduction}
It goes without doubt that the success of classical computers heavily relies on their ability to perform arithmetic evaluations. One might argue that the unique nature of quantum computers works better for fundamentally different approaches to information processing like the prominent variational quantum algorithms \cite{Cerezo2021}. However, the following example shows that once reliable hardware is available, quantum computers can surpass classical computers regarding the FLOPS metric\footnote{FLOPS stands for \textbf{fl}oating point \textbf{o}perations \textbf{p}er \textbf{s}econd, which is a measure for the speed of a computation device}.\\
To elaborate on this claim, consider the situation that an arbitrary quantum algorithm requires the multiplication of two numbers. For this, we assume the following partition of quantum registers:
\begin{itemize}
\item Factor 1 register
\item Factor 2 register
\item Target register
\item Miscellaneous register
\end{itemize}
The factor registers hold the information corresponding to each factor; the target register stores the multiplication results. The miscellaneous register is an $m$ qubit register which holds any other information that the algorithm produced so far.\\
We assume that, before the multiplication, the quantum computer is in the state
\begin{align}
\ket{\psi} = \sum_{i = 0}^{2^m-1} a_i \ket{x_i}\ket{y_i}\ket{0}\ket{i}.
\end{align}
Applying the multiplication circuit to this state let us estimate how many multiplications we can execute simultaneously:
\begin{align}
U_{\text{mult}} \ket{\psi} = \sum_{i = 0}^{2^m-1} a_i \ket{x_i}\ket{y_i}\ket{x_i\cdot y_i}\ket{i}.
\end{align}
Since we evaluated one multiplication per index $i$, we are in the best case at $2^m$ multiplications! The resource requirements of the multiplication circuit do not grow with $m$, thus we conclude that this situation can easily lead to the case, where multiplications cost less than a single gate! This could be regarded as an "inverse curse of dimensionality"\footnote{"Curse of dimensionality" is a term coined by mathematician Richard Bellman, which describes the exponential growth of computation and data resources for higher-dimensional data}.\\
Another interesting estimation we can make using this insight is to derive how many qubits are required to bring today's exascale ($10^{18}$ FLOPS) computing to the zettascale ($10^{21}$ FLOPS). For this, we assume an entangling gate speed of $\approx 10^{-6}\text{s}$ \cite{Schaefer2018} and a 32-bit multiplication circuit depth\footnote{To define the depth of a circuit $C$, we divide $C$ into a sequence of timesteps. During each time step, each qubit can execute at most one elementary gate. The circuit depth is then defined as the total amount of time steps required to complete the circuit. Note that this definition is highly dependent on the set of elementary gates. For our benchmarks we used the gate set $\{\text{CX}, \text{RZ}, \text{SX}\}$.} $\approx 10^4$ (see fig. \ref{performances_plots}) yielding a processing time of $\approx 10^{-2}\text{s}$ per multiplication circuit execution. This implies that we need a superposition of $10^{19}$ computational basis states, which is reached at $m = 65$ qubits. Together with the $32 \cdot 3 = 96$ qubits from the factor and target registers and an additional $\approx 32$ ancilla qubits (which are required to steeply increase multiplication evaluation speed), we arrive at $\approx 200$ qubits.\\
This is however highly dependent on how the quantum algorithm makes use of this multiplication. An example of an application that has no such advantage would be given by a simple calculator application for humans to use. Here, we can only multiply one pair of factors at the same time, implying no such speed up. Furthermore, in a realistic scenario, the need for reliable gates and stable qubits implies a quantum error correction scheme which will further increase the required qubit count and gate speed.\\
Nevertheless, this example demonstrates that once the right conditions are met, quantum computing and especially quantum arithmetic can truly supercharge humanity's information processing volume, rendering this field into a fascinating area of research.\\
This paper is structured as follows:\\
In what is left of section \ref{sec:introduction} we will give a short overview of the basics of classical arithmetic and how it translates to quantum computers.\\
In section \ref{sec:fundamentals} we will define and discuss the fundamentals which are necessary to construct the semi-boolean polynomial encoder. \\
Section \ref{sec:arithmetic} introduces arithmetic operations as applications of the previously constructed encoding circuit. The remainder of this section discusses how signed and non-integer values can be treated.\\
The following section \ref{sec:in_place_ops} constructs several methods to perform in-place operations (ie. addition, multiplication, polynomial evaluations etc.).\\
After this, in section \ref{sec:performance} we discuss several techniques to improve the performance of the semi-boolean polynomial encoder followed by a complexity estimation.\\
Finally, in section \ref{sec:critical_points} the developed methods are critically reviewed, followed by a summary of our results in section \ref{sec:conclusions}.

\subsection{Overview}
\label{sec:overview}
This section aims to provide the reader with some basic information about how arithmetic operations are traditionally performed on classical computers and how these methods translate to quantum computers.\\
Most of today's implementations of arithmetic are performed in the binary system. Not only does this suit the inner workings of a computer but it also minimizes the complexity of the elementary building blocks of arithmetic circuits. To perform an addition the most straight-forward approach is to deploy a chain of so-called \textit{full adders} which perform the basic operation of vertical addition in binary eg.:\\
\begin{center}
\begin{tabular}{ccccc}
  & 1 & 0 & 1 & 0\\
+ &   & 1 & 1 & 1\\
\tiny{1}&\tiny{1}&\tiny{1} & &\\
\hline
1 & 0 & 0 & 0 & 1\\ \\
\end{tabular}
\end{center}
We implement a full adder for every digit position of both summands, which determines the sum digit and the next carry digit depending on both the summand digits and the carry digit. An algorithmic formulation can be found in algorithm \ref{alg:addition}.
\begin{algorithm}
	\caption{\label{alg:addition}Addition} 
	\hspace*{\algorithmicindent}\textbf{Input:} Binary strings of the summands $(x_i)_{i \leq n}, (y_i)_{i \leq n}$\\
	\hspace*{\algorithmicindent}\textbf{Output:} Binary string of the sum $x + y$
	\begin{algorithmic}[1]
		\State $s$ = 0
		\State $c_{\text{in}} = 0$
		\For {$i$ in $(0,1,2.. n)$}
			\State $c_{\text{out}}, s_i = \text{FullAdder}(x_i, y_i, c_{\text{in}})$
			\State $c_{\text{in}} = c_{\text{out}}$
		\EndFor
		\State $s_{n+1} = c_{\text{out}}$
		\State \textbf{return} $s$
	\end{algorithmic} 
\end{algorithm}
Addition techniques like this are often called \textit{carry-ripple} because in this approach certain additions can create a ripple of carries that propagates to the most significant position. Note that for adding two $n$-bit integers $N_1, N_2$, we need to execute $n$ full adders, implying a complexity of $\mathcal{O}(\text{log}(N))$.\\
In base 2 the full adder can be represented by a truth table:\\
\begin{center}
\begin{tabular}{ccc|cc}
$x_i$ & $y_i$ & $c_{\text{in}}$ & $c_{\text{out}}$ & $s_i$\\
\hline
0&0&0&0&0\\
0&0&1&0&1\\
0&1&0&0&1\\
0&1&1&1&0\\
1&0&0&0&1\\
1&0&1&1&0\\
1&1&0&1&0\\
1&1&1&1&1\\ \\
\end{tabular}
\end{center}
Note that we can't infer the constellation of $x_i$ and $c_{\text{in}}$ in the case that we just know $y_i = 1, c_{\text{out}} = 1, s_i = 0$, implying $\ket{x_i}\ket{y_i}\ket{c_{\text{in}}} \rightarrow \ket{s_i}\ket{y_i}\ket{c_{\text{out}}}$ is not reversible. From this, we conclude that there is no (ancilla free) quantum circuit that can perform this operation\footnote{Since every quantum gate can be represented by a unitary (and therefore invertible) operator, any sequence of such gates has to be invertible}. This is a bit inconvenient because especially for multiplications (which will be covered shortly), having an in-place addition is vital. This problem was however successfully addressed by Cucarro et al. in \cite{Cuccaro2004}. For their approach, they realized that for $x,y \in \mathbb{N}$ the in-place addition $\ket{x}\ket{y}\ket{0}_{s_{n+1}} \rightarrow \ket{x}\ket{x+y}$ by itself is indeed reversible implying that the structure of successive full-adders might be the problem. Therefore their new approach consisted of calculating a set of intermediate truth values (with so-called MAJ gates) for each bit which are then again processed (in reversed order) by a sequence of so-called UMA gates resulting in a characteristic V-shape of the belonging circuits. Even though this circuit design is cheap in qubits and gates, the V-shape prevents parallel execution of most of the gates which results in sub-optimal circuit depth.\\ \\
Evaluating multiplications is simple once we have access to in-place addition. This is because the product of any binary number $y$ with another number $y$ with just a single 1 (for instance $y = (100)_2$) is simply $x$ but \textit{bit shifted}\footnote{A bit shift is an operation which moves a string of bits into a certain direction. Left shifts are denoted using the operator $\ll$, whereas right shifts are written as $\gg$. The amount of the shift is determined by the second operand. So for instance $(1010)_2 \ll 3 = (1010000)_2$} depending on where $y$ had it's 1:
\begin{align}
12 \cdot 4 = (1100)_2 \cdot (100)_2 = (1100)_2 \ll 2 = (110000)_2 = 48.
\end{align}
If we now want to relax the restriction of the second factor $y$ having only a single 1, we can write $y$ as a sum of single-1-numbers 
\begin{align}
y = a_0 + a_1 + .. a_k,
\end{align}
for instance $(1010)_2 = (1000)_2 + (10)_2$ and then calculate
\begin{align}
\label{multiplication_example_eq}
x \cdot y = \sum_{i = 0}^k x \cdot a_i.
\end{align}
Here we can see why an in-place adder is so important: Without it, we would have to store (and uncompute) at least $k$ different numbers (one for each iteration of the sum).\\
We rewrite eq. \ref{multiplication_example_eq} in a more algorithmic manner
\begin{align}
\label{algorithmic_mult_eq}
x  y = \sum_{i = 0}^{n} (x\ll i) y_i.
\end{align}
Where $n$ is the bit size of $y$, $y_i \in \mathbb{F}$ is the truth value of the $i$-th digit of $y$ in binary. From this, we formulate algorithm \ref{alg:multiplication}.

\begin{algorithm}

	\caption{\label{alg:multiplication} Multiplication 1} 
	\hspace*{\algorithmicindent}\textbf{Input:} Binary strings of the factors $(x_i)_{i \leq n}, (y_i)_{i \leq n}$\\
	\hspace*{\algorithmicindent}\textbf{Output:} Binary string of the product $x \cdot y$
	\begin{algorithmic}[1]
		\State $s = 0$
		\For {$i$ in $(0,1,2.. n)$}
			\If{$y_i$}:
				\State $s += (x \ll i)$
			\EndIf
		\EndFor

		\State \textbf{return} $s$
	\end{algorithmic} 
\end{algorithm}
We see that for the multiplication of two $n$-bit integers $N_1, N_2$ the algorithm requires $\mathcal{O}(n)$ additions, implying a complexity of $\mathcal{O}(\text{log}(N)^2)$.\\ 
Furthermore, we note that for the case of quantum computers, the bit shift doesn't have to be performed physically (ie. through swaps) - it is enough to rewire the in-place addition circuit. Regarding the conditional execution of the in-place addition, there are multiple possibilities: We could either turn every gate of the addition into its controlled version (for instance using \cite{shende2006}) or only control certain key gates of the adder as demonstrated in \cite{Thapliyal2021}. Even though much cheaper in gate overhead, this only works for certain kinds of addition circuits. A third possibility, which we used to make the multiplication performance of various adders comparable to our methods (see fig. \ref{performances_plots}), compromises the best of both worlds. For this, we use the well-known \cite{Thapliyal2021} identity\begin{align}
a - b = \overline{(\overline{a} + b)}
\end{align}
where the $\overline{a}$ is the bitwise negation of $a$. Since we can easily perform conditional bitwise negation using multiple CNOT gates, this allows for a versatile and efficient realization of carry-ripple adder-based multiplication circuits. To see how this can be used, we reformulate eq. \ref{algorithmic_mult_eq}:
\begin{equation}
\begin{aligned}
x y &= \sum_{i = 0}^{n} (x\ll i) y s_i\\
&= \sum_{i = 0}^{n} (x\ll i) \frac{1 - (-1)^{y_i}}{2}\\
&= \frac{1}{2}\left(\sum_{i = 0}^{n} (x\ll i) - \sum_{i = 0}^{n} (x\ll i) (-1)^{y_i}\right)\\
&= \frac{1}{2}\left(x \sum_{i = 0}^{n} 2^{i} - \sum_{i = 0}^{n} (x\ll i) (-1)^{y_i}\right)\\
&= \frac{1}{2}\left(x \frac{2^{n+1} - 1}{2-1} - \sum_{i = 0}^{n} (x\ll i) (-1)^{y_i}\right)\\
&= \left(x \ll (n + 1) - x - \sum_{i = 0}^{n} (x \ll i) (-1)^{y_i}\right) \gg 1
\end{aligned}
\end{equation}
From this, we give the modified version of algorithm \ref{alg:multiplication}:
\begin{algorithm}

	\caption{\label{alg:multiplication2} Multiplication 2} 
	\hspace*{\algorithmicindent}\textbf{Input:} Binary strings of the factors $(x_i)_{i \leq n}, (y_i)_{i \leq n}$\\
	\hspace*{\algorithmicindent}\textbf{Output:} Binary string of the product $x \cdot y$
	\begin{algorithmic}[1]
		\State $s = x \ll (n + 1)$
		\State $s -= x$
		\For {$i$ in $(0,1,2.. n)$}
			\If{$y_i$}:
				\State $s += (x \ll i)$
			\Else:
				\State $s -= (x \ll i)$
			\EndIf
		\EndFor

		\State \textbf{return} $s \gg 1$
	\end{algorithmic} 
\end{algorithm}

Even though far from complete, this concludes the short overview on standard methods for arithmetic evaluations. In the upcoming sections, we will see a very different approach, which has its own benefits and drawbacks compared to the methods we presented so far.

\section{Fundamentals}
\label{sec:fundamentals}
In this section, we will define the concept of semi-boolean polynomials and lay out the necessary techniques to encode their evaluation into circuits.
\subsection{Modular arithmetic}
A central part of the methods that are described below is modular arithmetic, therefore we will provide a short coverage of the belonging basics.\\ \\
For $x \in \mathbb{Z},\text{ }y \in \mathbb{N}$, the modulo operator \textbf{mod} maps to the smallest positive number $z$ such that
\begin{align}
z = x + jy,
\end{align}
where $j \in \mathbb{Z}$ is an integer. For instance, we have $7\text{ mod }3 = 1$ ($j$ would be equal to $2$ in this case).
Given an integer $n \in \mathbb{N}$, we can construct a ring\footnote{Simply put, a ring is a field without division.} by applying the modulus function to the set of integers. This ring is denoted as follows:
\begin{align}
\mathbb{Z}/n\mathbb{Z} := \mathbb{Z} \text{ mod } n.
\end{align}
For $a,b \in \mathbb{Z}/n\mathbb{Z}$ the arithmetic operations in $\mathbb{Z}/n\mathbb{Z}$ are realized by applying $\text{mod } n$ after their respective $\mathbb{Z}$-operations. For example:
\begin{align}
\text{mul}_{\mathbb{Z}/n\mathbb{Z}}(a, b) = \text{mul}_{\mathbb{Z}}(a, b)\text{ mod }n,\\
\text{add}_{\mathbb{Z}/n\mathbb{Z}}(a,b) = \text{add}_{\mathbb{Z}}(a, b)\text{ mod }n.
\end{align}
\subsection{Semi-Boolean Polynomials}
\label{sec:sbp_encoder}
Semi-boolean polynomials will play an important role in the course of this paper. The basic idea is to use a technique described in \cite{Gilliam2021groveradaptive} to encode a certain type of polynomial. Although being unnamed by the authors of \cite{Gilliam2021groveradaptive}, we will denote these polynomials as \textit{semi-boolean} to differentiate between polynomials with more general domains. A \textit{semi-Boolean polynomial} (from now on: \textit{SB-polynomial}) is a multivariate polynomial that has Boolean tuples as its domain but arbitrary real numbers as coefficients:
\begin{align}
\Omega : \mathbb{F}_2^n \rightarrow \mathbb{R}.
\end{align}
Furthermore, we define an \textit{integer SB-polynomial} as an SB-polynomial which has only integer coefficients. An example would be given by \ref{sbp_example_eq} where $x_0, x_1$ and $x_2$ denote boolean variables:
\begin{align}
\label{sbp_example_eq}
\Omega(x) = 4 x_0 x_2 - 3 x_1.
\end{align}
Integer SB-polynomial evaluations can be encoded into circuits i.e., for two registers (domain and image) of size $n,m \in \mathbb{N}$ respectively and a given integer SB-polynomial $\Omega$, this encoding acts as
\begin{align}
\label{sb_encoder_definition}
U_{\text{sbp}}(\Omega) \ket{x} \ket{0} = \ket{x}\ket{\Omega(x) \text{ mod } 2^m}.
\end{align}
How does this work? The idea is to construct the QFT of the state $\ket{\Omega(x)}$ followed by an inverse QFT. Encoding the Fourier-transform of $\ket{\Omega(x)}$ is advantageous because of the additive nature of phase gates.\\
We will first lay out the procedure of constructing the Fourier transformed state without any domain register. The encoding circuit will then follow by turning some of the gates into controlled operations.\\ \\
Suppose we want to encode the Fourier transform of the state $\ket{y}$, where $y$ is some integer. The first step is to initialize the image register into the state of uniform superposition by applying H gates on every qubit:
\begin{align}
\ket{s} = \frac{1}{\sqrt{2^m}}\sum_{k = 0}^{2^m-1} \ket{k}
\end{align}
We now define the gate
\begin{align}
U_{\mathrm{G}}(y) = \bigotimes_{i = 0}^{m-1} \text{P}_i\left(\frac{2\pi y 2^i}{2^m}\right),
\end{align}
where $\text{P}_i(\phi)$ is the parametrized phase gate, $\text{diag}(1,\text{exp}(i\phi))$ applied on the $i$-th qubit of the register. Applying $U_{\mathrm{G}}(y)$ to the uniform superposition state gives
\begin{equation}
\label{fourier_trafoed_y}
\begin{aligned}
&U_{\mathrm{G}}(y)\ket{s}\\
= &\frac{1}{\sqrt{2^m}} \sum_{k = 0}^{2^m-1} \left(\prod_{j = 0}^{m-1} \text{exp}\left(\frac{2\pi iy 2^jk_j}{2^m}\right)\right)\ket{k}\\
=&\frac{1}{\sqrt{2^m}} \sum_{k = 0}^{2^m-1} \text{exp}\left(\frac{2\pi i y}{2^m}\sum_{j = 0}^{m-1} 2^j k_j\right) \ket{k}\\
= &\frac{1}{\sqrt{2^m}} \sum_{k = 0}^{2^m-1} \text{exp}\left(\frac{2\pi i yk}{2^m}\right) \ket{k}.	
\end{aligned}
\end{equation}
Using the conventions from Nielsen \& Chuang chapter 5.1 \cite{nielsen00} this is the Fourier transform of $\ket{y}$ implying:
\begin{align}
\label{U_G_action}
\text{QFT}^\dagger U_{\mathrm{G}}(y)\mathrm{H}^{\otimes n} \ket{0} = \ket{y\text{ mod }2^{m}}.
\end{align}
We inserted the modulus because if  $2^{m}\leq y = (y\text{ mod }2^m) + j 2^{m}$ in eq. \ref{fourier_trafoed_y}, the phase corresponding to $j 2^{m}$ results in an integer multiple of $2\pi$ and therefore vanishes. Note that $U_{\mathrm{G}}$ is additive in the sense that
\begin{equation}
\begin{aligned}
& U_{\mathrm{G}}(y_1)U_{\mathrm{G}}(y_2)\\
=& \left(\bigotimes_{i = 0}^{m-1} \text{P}_i\left(\frac{2\pi y_1 2^i}{2^m}\right)\right) \left(\bigotimes_{i = 0}^{m-1} \text{P}_i\left(\frac{2\pi y_2 2^i}{2^m}\right)\right)\\
=& \bigotimes_{i = 0}^{m-1} \left( \text{P}_i\left(\frac{2\pi y_1 2^i}{2^m}\right) \text{P}_i\left(\frac{2\pi y_2 2^i}{2^m}\right)\right)\\
 = &\bigotimes_{i = 0}^{m-1} \text{P}_i\left(\frac{2\pi 2^i (y_1+y_2)}{2^m}\right).
\end{aligned}
\end{equation}
Hence, we can conclude that:
\begin{align}
\label{U_G_additivity}
U_{\mathrm{G}}(y_1) U_{\mathrm{G}}(y_2) = U_{\mathrm{G}}(y_1 + y_2).
\end{align}
The next step for encoding SB-polynomials is to turn $U_{\mathrm{G}}$ into controlled gates. If we for instance, control $U_{\mathrm{G}}(y)$ on the 0-th qubit of the domain register, we can abuse our notation and write $U_{\mathrm{G}}(yx_0)$ because if $x_0 = 0$, the gate is not executed, which is equivalent to applying $U_{\mathrm{G}}(0)$. If $x_0 = 1$, $U_{\mathrm{G}}(y)$ is executed which is also equivalent to $U_{\mathrm{G}}(yx_0)$. The same works with more factors of $x_i$ by controlling $U_{\mathrm{G}}$ on the corresponding qubits. We, therefore, are now able to encode SB-monomials. For instance,
\begin{align}
\text{QFT}^\dagger U_{\mathrm{G}}(4 x_0 x_1 x_2)\ket{x}\ket{s} = \ket{x}\ket{(4 x_0 x_1 x_2)\text{ mod } 2^m}.
\end{align}
The last step is now simple. To encode an arbitrary SB-polynomial
\begin{align}
\Omega(x) = \sum_{i = 0}^l a_i \prod_j x_{b_{ij}},
\end{align}
we execute the circuit
\begin{align}
U_{\text{sbp}}(\Omega) = \text{QFT}^\dagger \prod_{i = 0}^l U_{\mathrm{G}}\left( a_i \prod_j x_{b_{ij}} \right)\mathrm{H}^{\otimes n}
\end{align}
Using eq. \ref{U_G_additivity} we end up with
\begin{align}
\label{sbp_encoder_circuit}
U_{\text{sbp}}(\Omega) = \text{QFT}^\dagger U_{\mathrm{G}}(\Omega(x)) \mathrm{H}^{\otimes n},
\end{align}
which combined with eq. \ref{U_G_action} yields the functionality specified in eq. \ref{sb_encoder_definition}.\\
Note that we abused our notation of SB-polynomials here because usually the expression $\Omega(x)$ denotes a polynomial evaluated at some value $x$. Throughout the rest of this paper, if such an expression appears inside of the SB-polynomial encoder, the symbol $x$ denotes the register on which the $U_{\mathrm{G}}$ gates are controlled.

\begin{figure}
\scalebox{0.8}{
\begin{quantikz}
\lstick[wires = 3]{$\ket{x}$} 
\qw 	& \qw 		&\ctrl{3}			&\qw 				&\ctrl{3} 			&\qw 							& \qw \\
\qw 	& \qw 		&\qw				&\ctrl{4}			&\ctrl{3}	 		&\qw 							& \qw \\
\qw 	& \qw 		&\qw				&\ctrl{3} 			&\qw 				&\qw 							& \qw \\
\lstick[wires = 3]{$\ket{0}^{\otimes 3}$} 
\qw 	& \gate{\mathrm{H}} 	& \gate[3]{U_{\mathrm{G}}(1)} 	& \gate[3]{U_{\mathrm{G}}(2)}	&\gate[3]{U_{\mathrm{G}}(3)}	&\gate[3]{\text{QFT}^\dagger} 	& \qw \\
\qw 	& \gate{\mathrm{H}}	&	 				& \qw				&\qw	 			& 								& \qw\\
\qw 	& \gate{\mathrm{H}} 	& 					& \qw 				& \qw		 		&								& \qw 
\end{quantikz}}
\caption{SB-polynomial encoding of $\Omega(x) = x_0  + 2 x_1 x_2 + 3 x_0 x_1$.}
\end{figure}
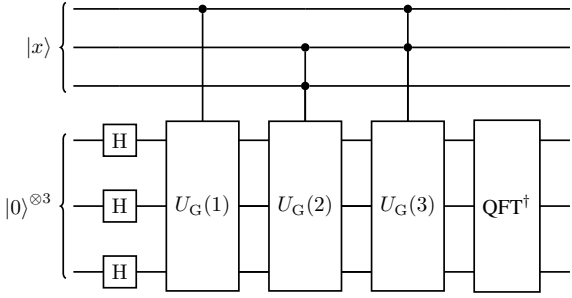

\section{Arithmetic}
\label{sec:arithmetic}
With the mathematical fundamentals and the SB-polynomial encoder at hand, we can now start constructing arithmetic circuits.
\subsection{Unsigned integer arithmetic}
Consider now 3 quantum registers with size $n_1, n_2, m \in \mathbb{N}$ and the integers $x < 2^{n_1}$ and $y < 2^{n_2}$.
The goal of this section is to construct circuits that perform the following operations:	
\begin{align}
U_{\text{add}} \ket{x}\ket{y} \ket{0} = \ket{x}\ket{y} \ket{(x+y)\text{ mod }2^m},\\
U_{\text{sub}} \ket{x}\ket{y} \ket{0} = \ket{x}\ket{y} \ket{(x-y)\text{ mod }2^m},\\
U_{\text{mult}} \ket{x}\ket{y} \ket{0} = \ket{x}\ket{y} \ket{(xy)\text{ mod }2^m}.
\end{align}
Additionally, we will construct a technique that can encode an arbitrary multivariate integer coefficient polynomial (not semi-boolean!) $p(x,y,z,..)$, such that
\begin{equation}
\begin{aligned}
\label{poly_encoder}
&U_{\text{poly}}(p) \ket{x} \ket{y} \ket{z} .. \ket{0}\\
=& \ket{x} \ket{y} \ket{z} .. \ket{p(x,y,z,..)\text{ mod }2^m}.
\end{aligned}
\end{equation}
The key insight in finding these circuits is that the evaluation of the binary representation of $x$ to its value is an SB-polynomial:
\begin{align}
\label{binary_expansion}
x = \sum_{k = 0}^{n_1 - 1} 2^k x_k.
\end{align}
We formalize this idea by defining the $n$-bit unsigned integer encoding SB-polynomial
\begin{equation}
\label{sb_poly_binary_exp}
\begin{aligned}
\Omega^n_{\text{us}}: \mathbb{F}^n_2 &\rightarrow \mathbb{N},\\
(x_0, x_1, .. ,x_n) &\rightarrow \sum_{k = 0}^{2^n-1} 2^k x_k.
\end{aligned}
\end{equation}
The next step is to write the arithmetic evaluations in terms of these SB-polynomials. This is not hard: The SB-polynomial for the addition is simply the sum of the corresponding SB-polynomials.
From this, we conclude that the conversion from the binary representations of $x$ and $y$ to the value of $x+y$ can also be written as a semi-Boolean polynomial:
\begin{equation}
\begin{aligned}
x + y &= \sum_{k = 0}^{n_1-1} 2^k x_k +  \sum_{k = 0}^{n_2-1} 2^k y_k\\
&= \Omega_{\text{us}}^{n_1}(x_0, x_1, .. ,x_n) + \Omega_{\text{us}}^{n_2}(y_0, y_1, .. ,y_n).
\end{aligned}
\end{equation}
Similarly, the difference and the product can also be written as the evaluation of SB polynomials. Our arithmetic circuits therefore simply arise from substituting the representation from eq. \ref{binary_expansion} into the arithmetic operation and afterward encoding the resulting semi-Boolean polynomial by using the encoding circuit we constructed in section \ref{sec:fundamentals} i.e.:
\begin{align}
U_{\text{add}} = U_{\text{sbp}}(\Omega_{\text{us}}^{n_1}(x) + \Omega_{\text{us}}^{n_2}(y)),\\
U_{\text{sub}} = U_{\text{sbp}}(\Omega_{\text{us}}^{n_1}(x) - \Omega_{\text{us}}^{n_2}(y)),\\
U_{\text{mult}} = U_{\text{sbp}}(\Omega_{\text{us}}^{n_1}(x)\Omega_{\text{us}}^{n_2}(y)).
\end{align}
The circuits that come out of this algorithm have been known as Fourier-arithmetic and have been studied in \cite{Draper2000}, \cite{Ruiz-Perez2017}.\\
For an arbitrary, multivariate, integer coefficient polynomial $p: \mathbb{R}^k \rightarrow \mathbb{R}$, we can use that $p(\Omega_{\text{us}}^{n_1}(x), \Omega_{\text{us}}^{n_2}(y),\Omega_{\text{us}}^{n_3}(z),..)$ is again an SB-polynomial which can therefore be encoded using eq. \ref{sb_encoder_definition}:
\begin{equation}
\label{arbtitrary_polynomial_encoder}
U_{\text{poly}}(p) = U_{\text{sbp}}(p(\Omega_{\text{us}}^{n_1}(x), \Omega_{\text{us}}^{n_2}(y),\Omega_{\text{us}}^{n_3}(z),..)).
\end{equation}
Note that this circuit evaluates $p$ without the computation of any intermediate values making it very light on the qubit count.\\

\subsection{Signed Integer Arithmetic}
\label{signed_integer_arithmetic}
The above discussions construct the simple case of executing basic arithmetic on unsigned integers. However, it cannot deal with signed integers and the belonging operations. To treat negative values, the first question which arises is how to represent them. The most basic approach would be to add a sign bit $x_s$ which then indicates the sign: 
\begin{align}
\text{sign}(x) = 
\begin{cases}
1 & \text{ if } x_s = 0\\
-1 & \text{ if } x_s = 1
\end{cases}.
\end{align}
In principle this approach might be a suitable approach for our techniques, however, a more detailed analysis uncovers many disadvantages for implementation on a quantum computer. For the multiplication of two numbers, is be straightforward to determine the sign bit of the result. The sign bit of the product can be determined by applying CNOT gates controlled on the sign bits of the factors. For sums the situation is already much more complicated: Imagine adding two numbers $x,y$, where one is positive and the other is negative. The sign of the sum depends on whether the $|x| \geq |y|$. Even though a comparison like this would in principle be possible to implement on a quantum computer \cite{Oliveira2007}, it is still very unwieldy. Furthermore, evaluating polynomials like eq. \ref{poly_encoder} would not be possible anymore, because every single monomial needs to be computed separately to determine its absolute value.\\
Given the above considerations, we pick a different encoding for negative numbers, which makes use of the modular structure\footnote{Modular structure means that every arithmetic operation is performed mod $2^n$ because of eq. \ref{U_G_action}} of the unsigned integer arithmetic. With this approach, there is no need to construct extra circuits for the signed integer arithmetic and every operation for the unsigned case, directly translates to the signed case. The one thing which does change is the de/encoding of the corresponding numbers. To see how this works, we begin by defining the \textit{n-bit signed integer encoding ring isomorphism}
\begin{equation}
\label{iso_def}
\begin{aligned}
[.]_n:-\mathbb{Z}/2^{n}\mathbb{Z} \cup \mathbb{Z}/2^{n}\mathbb{Z} \rightarrow \mathbb{Z}/2^{n+1}\mathbb{Z},\\
x \rightarrow 
\begin{cases} 
x &\text{if } x \geq 0\\
2^{n+1} - |x| &\text{else}
\end{cases}.
\end{aligned}
\end{equation}
An intuitive understanding of what this does can be gained from viewing $\mathbb{Z}/12\mathbb{Z}$ as a clock-face. In this context, the numbers from $-6$ to $6$ are being represented on the clock. Positive values like $3$ or $4$ are being represented by themselves, whereas negative values are being represented by moving the pointer counterclockwise. We give some examples:
\begin{align}
[0]_3 &= 0,\\
[3]_3 &= 3,\\
[-4]_3 &= 16 - 4 = 12.
\end{align}
If we now want to execute an arithmetic operation on a signed integer $x$, we encode $[x]_n$ and execute the unsigned operation. If at some point we want to extract the results, we measure and simply apply the inverse of eq. \ref{iso_def}
\begin{align}
\label{inv_iso}
[y]_n^{-1} = \begin{cases}
y &\text{ if } y < 2^n\\
y - 2^{n+1} &\text{ else }
\end{cases},
\end{align}
for this formalism to recover the signed result. We will now provide proof of why this works for handling signs. What is required to show is, that for two $n$-bit signed integers $x,y$, regardless of the constellation of signs, the following identities hold:
\begin{align}
[x + y]_n \text{ mod } 2^{n+1} = ([x]_n + [y]_n) \text{ mod } 2^{n+1},\\
[x - y]_n \text{ mod } 2^{n+1} = ([x]_n - [y]_n) \text{ mod } 2^{n+1},\\
\label{signed_product_homomorphism}
[xy]_n \text{ mod } 2^{n+1} = [x]_n [y]_n \text{ mod } 2^{n+1}.
\end{align}
Let $x<0$ and $y\geq 0$ such that the multiplication does not result in an overflow, i.e. $|xy|$ can still be represented by an $n$-bit integer - otherwise increase $n$:
\begin{equation}
\begin{aligned}
[x]_n [y]_n\text{ mod } 2^{n+1} &= (2^{n+1} - |x|)y \text{ mod } 2^{n+1}\\
&= (2^{n+1}y - |x|y) \text{ mod } 2^{n+1}\\
&= (-|x|y) \text{ mod } 2^{n+1}\\
&= (2^{n+1} - |x|y) \text{ mod } 2^{n+1}\\
&= (2^{n+1} - |xy|) \text{ mod } 2^{n+1}\\
&= [xy]_n \text{ mod } 2^{n+1}.
\end{aligned}
\end{equation}
Here, we used that $2^{n+1}y$ is either $0$ or an integer multiple of $2^{n+1}$ and thus vanishes after applying the modulus. In the following step, we added $2^{n+1}$ knowing that it also vanishes because of the modulus. The case where both $x$ and $y$ are non-negative is trivial because this is just unsigned integer multiplication. What remains to be shown, is the case where both $x$ and $y$ are negative:
\begin{equation}
\begin{aligned}
&[x]_n [y]_n\text{ mod } 2^{n+1}\\
 &= (2^{n+1} - |x|)(2^{n+1} - |y|) \text{ mod } 2^{n+1}\\
&= (2^{2(n+1)}-2^{n+1}(|x| + |y|) + |x||y|)\text{ mod } 2^{n+1}\\
&= |x||y|\text{ mod } 2^{n+1}\\
&= xy \text{ mod } 2^{n+1}\\
&= [xy]_n \text{ mod } 2^{n+1}.
\end{aligned}
\end{equation}
In some situations, the following formula comes in more handy than eq. \ref{iso_def}:
\begin{align}
[x]_n = x \text{ mod } 2^{n+1}.
\end{align}
One such situation is proving the addition-equivalent of the above identities. Let $x,y$ be arbitrary signed $n$-bit integers
\begin{equation}
\begin{aligned}
[x + y]_n \text{ mod } 2^{n+1} &= (x + y) \text{ mod } 2^{n+1}\\
& = (x \text{ mod } 2^{n+1} + y \text{ mod } 2^{n+1}) \text{ mod } 2^{n+1}\\
& = ([x]_n + [y]_n) \text{ mod } 2^{n+1}.
\end{aligned}
\end{equation}
The case of subtraction works the same. With this at hand, we can describe the signed integer arithmetic circuits.\\
For this, we assume for now, that all participating registers have size $n$. We will shortly lift this restriction. Furthermore, let $x,y,z$ again be signed $n$-bit integers
\begin{equation}
\begin{aligned}
&U_{\text{add}} \ket{[x]_n} \ket{[y]_n} \ket{0}\\
&= \ket{[x]_n} \ket{[y]_n} \ket{([x]_n + [y]_n) \text{ mod } 2^{n+1}}\\
& = \ket{[x]_n} \ket{[y]_n} \ket{[x + y]_n \text{ mod } 2^{n+1}}.
\end{aligned}
\end{equation}
Similarly,	 for the multiplication, subtraction, and polynomial encoding we get
\begin{equation}
\begin{aligned}
&U_{\text{sub}} \ket{[x]_n} \ket{[y]_n} \ket{0} = \ket{[x]_n} \ket{[y]_n} \ket{[x-y]_n \text{ mod } 2^{n+1}},\\
&U_{\text{mult}} \ket{[x]_n} \ket{[y]_n} \ket{0} = \ket{[x]_n} \ket{[y]_n} \ket{[xy]_n \text{ mod } 2^{n+1}},\\
&U_{\text{poly}}(p) \ket{[x]_n} \ket{[y]_n} \ket{[z]_n} ..\ket{0}, \\
& = \ket{[x]_n} \ket{[y]_n} \ket{[z]_n} .. \ket{[p(x,y,z..)]_n \text{ mod } 2^{n+1}}.
\end{aligned}
\end{equation}
To provide an explicit example suppose we want to multiply the 3-bit signed integers $-3$ and $2$. Then we first encode $[-3]_3 = 13$ and $[2]_3 = 2$ and multiply using the unsigned circuits
\begin{equation}
\begin{aligned}
U_{\text{mult}} \ket{[-3]_n} \ket{[2]_n} \ket{0} &= U_{\text{mult}} \ket{13} \ket{2} \ket{0}\\
&= \ket{13} \ket{2} \ket{26\text{ mod }2^{3+1}}\\
&= \ket{13} \ket{2} \ket{10}.\\
\end{aligned}
\end{equation}
Measuring the third register yields the measurement result $10$. To complete the example calculation, we apply the inverse eq. \ref{inv_iso}
\begin{align}
[10]_3^{-1} = -6.
\end{align}
\subsection{Signed Integer Arithmetic with Flexible Register Sizes}
Even though sufficient when provided with large numbers of quantum resources, the method described in section \ref{signed_integer_arithmetic} has the disadvantage that domain and image registers of the arithmetic operations all need to have the same size, for the modular arithmetic to work out correctly. This consumes a lot of unnecessary resources, if we for instance know that one factor is always much smaller than the other, but also creates an even more practical problem: Because the result register size is limited, we have no way of avoiding overflow. With an $n$-bit multiplication, the largest possible outcome is $(2^n-1)^2 = 2^{2n} - 2^{n+1} + 1$ implying the result register size needs to be at least $2n$-bit to prevent overflow.\\ \\
Both points above justify spending some effort on lifting the restriction of fixed register sizes. Our approach here is to derive a conversion rule for the conversion between different signed integer isomorphisms (e.g. $[x]_3 \rightarrow [x]_4$) and formulate this as an SB-polynomial. In order to acquire the conversion rule, let $x$ be an $n$-bit signed integer and $n<m \in \mathbb{N}$. If $x\geq 0$, we simply have $[x]_n = x = [x]_m$. If $x < 0$,
\begin{equation}
\begin{aligned}
[x]_m &= 2^{m+1} - |x|\\
&= 2^{m+1} - 2^{n+1} + 2^{n+1} - |x|\\
&= 2^{m+1} - 2^{n+1} + [x]_n.
\end{aligned}
\end{equation}
Summarizing both cases, we obtain
\begin{align}
\label{iso_conversion}
[x]_{m} = \begin{cases}
[x]_n &\text{if } [x]_n < 2^n\\
2^{m+1}-2^{n+1} + [x]_n &\text{else}
\end{cases}.
\end{align}
This conversion can be expressed as an SB-polynomial:
\begin{align}
\label{sb_poly_image_extension}
\Omega_{\text{IE}}^{n,m}(x) = (2^{m+1}-2^{n+1}) x_{n} + \sum_{i = 0}^{n} 2^i x_i 
\end{align}
where $\text{IE}$ stands for \textit{image extension}. 
In order to derive a relationship we will need soon, we insert $[x]_n$ in eq. \ref{sb_poly_image_extension}:
\begin{equation}
\begin{aligned}
&\Omega_{\text{IE}}^{n,m}([x]_n) \\
&= (2^{m+1}-2^{n+1}) ([x]_n)_{n} + \sum_{i = 0}^{n} 2^i ([x]_n)_i\\
&= \begin{cases}
\sum_{i = 0}^{n} 2^i ([x]_n)_i & \text{ if } ([x]_n)_{n} = 0\\
2^{m+1} - 2^{n+1} +  \sum_{i = 0}^{n} 2^i ([x]_n)_i& \text{ if } ([x]_n)_{n} = 1
\end{cases}\\
&= \begin{cases}
[x]_n & \text{ if } [x]_n < 2^n\\
2^{m+1} - 2^{n+1} + [x]_n & \text{ else }
\end{cases}\\
\label{image_expanded_polynomial}
&= [x]_{m}.
\end{aligned}
\end{equation}
We now use the SB-polynomial from eq. \ref{sb_poly_image_extension} instead of eq. \ref{sb_poly_binary_exp} to encode signed arithmetic on registers with size $n_1, n_2, m$
\begin{align}
U_{\text{add}}^{n_1, n_2, m} = U_{\text{sbp}}(\Omega_{\text{IE}}^{n_1,m}(x) + \Omega_{\text{IE}}^{n_2, m}(y)),\\
U_{\text{sub}}^{n_1, n_2, m} = U_{\text{sbp}}(\Omega_{\text{IE}}^{n_1,m}(x) - \Omega_{\text{IE}}^{n_2, m}(y)),\\
U_{\text{mult}}^{n_1, n_2, m} = U_{\text{sbp}}(\Omega_{\text{IE}}^{n_1,m}(x)\Omega_{\text{IE}}^{n_2, m}(y)).
\end{align}
Applying the multiplication to the state $\ket{[x]_{n_1}} \ket{[y]_{n_2}} \ket{0}$ then yields
\begin{equation}
\label{flexible_register_size_arthm_evaluation}
\begin{aligned}
&U_{\text{mult}}^{n_1, n_2, m} \ket{[x]_{n_1}} \ket{[y]_{n_2}} \ket{0}\\
&\underset{\text{eq.} \ref{sb_encoder_definition}}{=} \ket{[x]_{n_1}} \ket{[y]_{n_2}} \ket{\Omega_{\text{IE}}^{n_1,m}([x]_{n_1})\Omega_{\text{IE}}^{n_2, m}([x]_{n_1}) \text{ mod } 2^{m+1}}\\
&\underset{\text{eq.} \ref{image_expanded_polynomial}}{=} \ket{[x]_{n_1}} \ket{[y]_{n_2}} \ket{[x]_m [y]_m \text{ mod } 2^{m+1}}\\
&\underset{\text{eq.} \ref{signed_product_homomorphism}}{=} \ket{[x]_{n_1}} \ket{[y]_{n_2}} \ket{[xy]_m}
\end{aligned}
\end{equation}
In the first equality, we used the functionality of the SB-polynomial encoder eq. \ref{sb_encoder_definition}. The second equality is eq. \ref{image_expanded_polynomial} and the last equality is eq. \ref{signed_product_homomorphism} and using the fact that we assume suited register sizes to prevent overflow for the given multiplication.\\
Eq. \ref{flexible_register_size_arthm_evaluation} shows that the concept of the image extended SB-polynomial $\Omega_\text{IE}^{n,m}$ yields the desired results for flexible register size arithmetic.

\subsection{Floating-Point Arithmetic}
The last step for full floating-point arithmetic is encoding for the arithmetic of non-integers. For this, we will use the following representation: Let $x \in \mathbb{Q}$ be a rational such that its binary representation is finite (ie. not something like $0.\overline{01}$). Then there exist integers $k, l \in \mathbb{Z}$ such that
\begin{align}
x = \pm \sum_{i = k}^{l} 2^i x_i,
\end{align}
where $\forall k \leq i \leq l : x_i \in \mathbb{F}_2 $. We write this as
\begin{align}
x = \underline{x} 2^{k},
\end{align}
where $\underline{x}$ is a signed $n = k + l$ bit integer. We call $k$ the exponent and $\underline{x}$ the mantissa of $x$. Even though it would be straightforward to encode both mantissa and exponent as quantum variables, performing arithmetic operations on such a \textit{bi-quantum encoding} does not turn out as simple. To perform the addition of two bi-quantum encoded floats, the mantissa would have to be bit shifted depending on the state of the exponent register. Such a circuit could in principle be constructed using Fredkin gates\footnote{A Fredkin gate is a controlled SWAP gate} in combination with an incrementor gate\footnote{An incrementor gate is a gate that when applied to an $n$-qubit register, performs the operation $\ket{x} \rightarrow \ket{x + 1\text{ mod }2^n}$} however it would be very unwieldy. The data-format we are about to construct does not need such quantum-conditioned bit shifting because only the mantissa is quantum - the exponent is a classically known number. We call this property \textit{mono-quantum encoding}.\\
We extend the $[.]$ notation:
\begin{align}
\label{non_int_iso}
[x]_n^k := [\underline{x}]_n = [2^{-k} x]_n.
\end{align}
We will now prove an identity that is required later (in eq. \ref{example_float_eval}). Let $0\leq j \in \mathbb{N}_0$ we then have
\begin{equation}
\label{pulling_eq}
\begin{aligned}
2^j [x]_n \text{ mod } 2^{n+1} &= \begin{cases}
2^j x \text{ mod } 2^{n+1}&\text{if } x \geq 0\\
(2^{j+n+1} - 2^j x )\text{mod } 2^{n+1} &\text{else}
\end{cases}\\
&=\begin{cases}
2^j x \text{ mod } 2^{n+1} &\text{if } x \geq 0\\
(2^{n+1} - 2^j x )\text{ mod } 2^{n+1} &\text{else}
\end{cases}\\
&= [2^j x]_n \text{mod } 2^{n+1}.
\end{aligned}
\end{equation}
Simply plugging the SB-polynomials of the mantissa into the SB-polynomial encoder (as we did in the previous sections) however, does not suffice because we only defined the SB-polynomial encoder for integer coefficients. In order to still be able to encode non-integer values, our solution is to bit shift the polynomial (ie. multiply by a power of two) such that all its coefficients are integers. This bit shift will be reversed when decoding a measurement outcome. For this, we introduce a new notation for the SB-polynomial encoder. Consider an arithmetic operation writing into a target register with exponent $k_0$. The bit shifted SB-polynomial encoder is now defined as
\begin{align}
\label{bit_shifted_sbp_encoder}
U_\text{sbp}^{k_0}(\Omega) := U_\text{sbp}(2^{-k_0}\Omega).
\end{align}
This allows us to conveniently write down the encoding of the mantissa into the target register. Similarly, we extend the notation of the encoding SB-polynomial to account for the fact that they have to turn the boolean array of the mantissa $\underline{x}$ into the actual value $x$. For this, consider a domain register $x$ with size $n_1$ and exponent $k_1$:
\begin{align}
\Omega_{\text{IE}}^{n_1,m,k_1}(\underline{x}) := 2^{k_1} \Omega_{\text{IE}}^{n_1,m}(\underline{x}).
\end{align}
Using these definitions, it is once again possible to write down arithmetic circuits for floating-point operations. For this we assume that the relevant registers have sizes $n_0, n_1, n_2, ...$ and exponents $k_0, k_1, k_2,..$. The resulting arithmetic circuits can be summarized as follows:
\begin{align}
\label{floating_point_arithmetic}
U_{\text{add}} = U_{\text{sbp}}^{k_0}(\Omega_{\text{IE}}^{n_1,n_0,k_1}(\underline{x}) + \Omega_{\text{IE}}^{n_2,n_0,k_2}(\underline{y})),\\
U_{\text{sub}} = U_{\text{sbp}}^{k_0}(\Omega_{\text{IE}}^{n_1,n_0,k_1}(\underline{x}) - \Omega_{\text{IE}}^{n_2,n_0,k_2}(\underline{y})),\\
U_{\text{mult}} = U_{\text{sbp}}^{k_0}(\Omega_{\text{IE}}^{n_1,n_0,k_1}(\underline{x})\Omega_{\text{IE}}^{n_2,n_0,k_2}(\underline{y})).
\end{align}
Note that the target registers have certain requirements regarding their exponents, in order to ensure that only integer SB-polynomials are handed to the non-bit shifted SB-polynomial encoder. For instance, if we want to multiply two numbers with exponents $k_1, k_2$, the smallest possible outcome is $2^{k_1 + k_2}$. Hence, the target register needs to have an exponent of $k_0 \leq k_1 + k_2$. Anything above this threshold can not support the result of the $2^{k_1} \times 2^{k_2}$ multiplication.\\
To summarize, the following rules are valid for shape determination:
\begin{itemize}
\item Addition: $k_0 \leq \text{min}(k_1, k_2)$
\item Multiplication: $k_0 \leq k_1 + k_2$
\end{itemize}
These techniques also generalize to the arbitrary polynomial circuit eq. \ref{arbtitrary_polynomial_encoder}. Note that we can lift the restriction of integer coefficients as long as the bit shifted SB-polynomial
\begin{align}
2^{-k_0} p(\Omega_{\text{us}}^{n_1, n_0, k_1}(x), \Omega_{\text{us}}^{n_2, n_0, k_2}(y),\Omega_{\text{us}}^{n_3, n_0, k_3}(z),..))
\end{align}
contains only integer coefficients.\\
The following is an example for a floating-point arithmetic circuit evaluation for the multiplication of two floats with different data shapes:
\begin{equation}
\label{example_float_eval}
\begin{aligned}
&U_{\text{mult}} \ket{[x]_{n_1}^{k_1}} \ket{[y]_{n_2}^{k_2}} \ket{0} = U_{\text{mult}} \ket{[2^{-k_1}x]_{n_1}} \ket{[2^{-k_2}y]_{n_2}} \ket{0}\\
&= \ket{[2^{-k_1}x]_{n_1}} \ket{[2^{-k_2}y]_{n_2}} \\
&\ket{2^{-k_0}\Omega_{\text{IE}}^{n_1,n_0,k_1}([2^{-k_1}x]_{n_1})\Omega_{\text{IE}}^{n_2,n_0,k_2}([2^{-k_2}y]_{n_2})\text{ mod } 2^{n_0+1}}\\
&=\ket{[x]_{n_1}^{k_1}} \ket{[y]_{n_2}^{k_2}} \ket{2^{-k_0+k_1+k_2}[2^{-k_1}x]_{n_0}[2^{-k_2}y]_{n_0}\text{ mod } 2^{n_0+1}}\\
&=\ket{[x]_{n_1}^{k_1}} \ket{[y]_{n_2}^{k_2}} \ket{2^{-k_0+k_1+k_2}[2^{-k_1-k_2}xy]_{n_0}\text{ mod } 2^{n_0+1}}\\
&=\ket{[x]_{n_1}^{k_1}} \ket{[y]_{n_2}^{k_2}} \ket{[2^{-k_0}xy]_{n_0}\text{ mod } 2^{n_0+1}}\\
&=\ket{[x]_{n_1}^{k_1}} \ket{[y]_{n_2}^{k_2}} \ket{[xy]_{n_0}^{k_0}\text{ mod } 2^{n_0+1}}.
\end{aligned}
\end{equation}
As usual, we assume that the register sizes are chosen that neither an overflow nor an underflow is happening. Let's recapitulate what identities were used here
\begin{enumerate}
\item Application of definition of $[.]$ for non-integers eq. \ref{non_int_iso}
\item The bit shifted SB-polynomial encoder eq. \ref{bit_shifted_sbp_encoder}
\item Evaluation of the image extended SB-polynomials eq. \ref{image_expanded_polynomial}
\item Multiplicative morphism property of $[.]$ eq. \ref{signed_product_homomorphism}
\item Applying eq. \ref{pulling_eq} with\\$k_0 \leq k_1 + k_2 \Leftrightarrow -k_0 + k_1 + k_2 \geq 0$
\item Again eq. \ref{non_int_iso} (this time backward)
\end{enumerate}

\section{In-place operations}
\label{sec:in_place_ops}

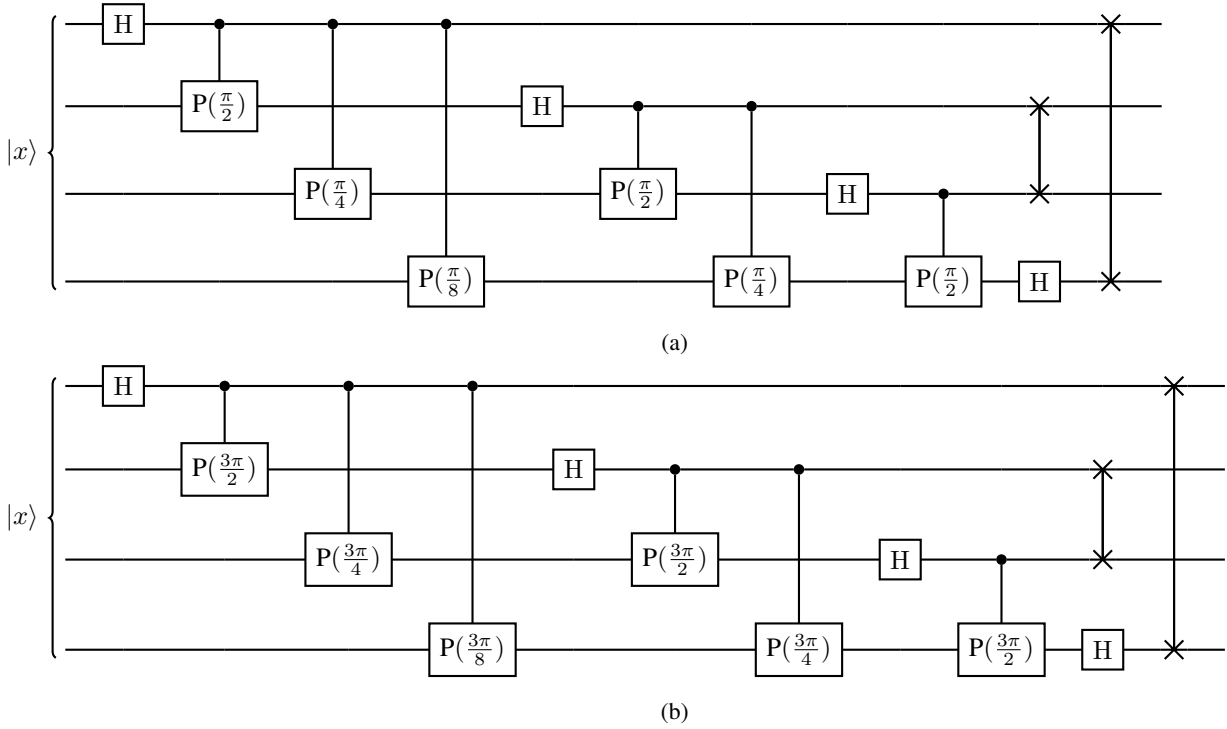
\begin{figure*}
\begin{center}
\scalebox{1}{
\begin{subfigure}{\textwidth}
\begin{quantikz}
\lstick[wires = 4]{$\ket{x}$}
\qw 	& \gate{\mathrm{H}}	&\ctrl{1}								&\ctrl{2}						&\ctrl{3} 						&\qw 		&\qw 									&\qw							&\qw		&\qw							&\qw		&\swap{3}	&\qw\\
\qw 	& \qw 		&\gate{\text{P}(\frac{\pi}{2})}			&\qw							&\qw		 					&\gate{\mathrm{H}}	&\ctrl{1}								&\ctrl{2}						&\qw		&\qw							&\swap{1}	&\qw		&\qw\\
\qw 	& \qw 		&\qw									&\gate{\text{P}(\frac{\pi}{4})}	&\qw							&\qw 		&\gate{\text{P}(\frac{\pi}{2})}			&\qw							&\gate{\mathrm{H}}	&\ctrl{1}						&\swap{-1}	&\qw		&\qw\\
\qw 	& \qw 		&\qw									&\qw 							&\gate{\text{P}(\frac{\pi}{8})} &\qw 		&\qw									&\gate{\text{P}(\frac{\pi}{4})}	&\qw		&\gate{\text{P}(\frac{\pi}{2})}	&\gate{\mathrm{H}}	&\swap{-3}	&\qw
\end{quantikz}
\caption{\label{regular_qft_circuit}}
\end{subfigure}}
\scalebox{1}{
\begin{subfigure}{\textwidth}
\begin{quantikz}
\lstick[wires = 4]{$\ket{x}$}
\qw 	& \gate{\mathrm{H}}	&\ctrl{1}								&\ctrl{2}						&\ctrl{3} 						&\qw 		&\qw 									&\qw							&\qw		&\qw							&\qw		&\swap{3}	&\qw\\
\qw 	& \qw 		&\gate{\text{P}(\frac{3\pi}{2})}		&\qw							&\qw		 					&\gate{\mathrm{H}}	&\ctrl{1}								&\ctrl{2}						&\qw		&\qw							&\swap{1}	&\qw		&\qw\\
\qw 	& \qw 		&\qw									&\gate{\text{P}(\frac{3\pi}{4})}&\qw							&\qw 		&\gate{\text{P}(\frac{3\pi}{2})}		&\qw							&\gate{\mathrm{H}}	&\ctrl{1}						&\swap{-1}	&\qw		&\qw\\
\qw 	& \qw 		&\qw									&\qw 							&\gate{\text{P}(\frac{3\pi}{8})}&\qw 		&\qw									&\gate{\text{P}(\frac{3\pi}{4})}&\qw		&\gate{\text{P}(\frac{3\pi}{2})}&\gate{\mathrm{H}}	&\swap{-3}	&\qw
\end{quantikz}
\caption{\label{in_place_mul}}
\end{subfigure}}
\caption{\label{in_place_mult_cirquits} \textbf{\ref{regular_qft_circuit}} The circuit for the reqular QFT. \textbf{\ref{in_place_mul}} The circuit for QFT in-place multiplication with $a = 3$.}
\end{center}
\end{figure*}

A significant disadvantage of the methods described in section \ref{sec:arithmetic} is that we always need a new register to store our results. This is the case for many quantum operations, in order to ensure the reversibility of the circuits. However, certain arithmetic operations are reversible by themselves and in principle it should be possible to implement them in-place i.e., operating only on the domain register and not requiring a target register. In order to reduce the complexity of the equations in this section, the presentation focuses on the unsigned integers and omits the variety of indices relating to the register sizes and exponents. However, all of the considerations in this section can also be performed with full floating-point arithmetic.\\
\subsection{In-place Addition}
Since modular addition is a reversible operation, implementing it in-place poses no major challenge. Indeed, it can be easily recognized that the initial H gates of the SB-polynomial encoder can be interpreted as a streamlined\footnote{"Streamlined" in this context means a more efficient version which is only valid if the input state is $\ket{0}$} version of a QFT that only acts on the $\ket{0}$ state
\begin{equation}
\begin{aligned}
(H\ket{0})^{n \otimes}  &= \frac{1}{\sqrt{2^n}} \sum_{k = 0}^{2^n-1} \ket{k}\\
&= \frac{1}{\sqrt{2^n}} \sum_{k = 0}^{2^n-1} \text{exp}\left(\frac{2\pi i 0 k}{2^{n}} \right)\ket{k} \\
&=  \text{QFT} \ket{0}.
\end{aligned}
\end{equation}
Simply replacing $\mathrm{H}^{n\otimes}$ in the SB-polynomial encoder (eq. \ref{sbp_encoder_circuit}) with the general form of a QFT yields the desired result
\begin{align}
U_{\text{sbp}}^{\text{in-place}}(\Omega(x)) = \text{QFT}^{\dagger} U_{\mathrm{G}}(\Omega(x)) \text{QFT}.
\end{align}
Here, the Fourier-transformations are applied onto the register, where the in-place operation is supposed to happen. Basically, everything discussed before - regarding sign and exponent treatment - now directly translates to this setting:
\begin{align}
U_{\text{add}}^{\text{in-place}} &= U_{\text{sbp}}^{\text{in-place}} (\Omega(x)),\\
U_{\text{sub}}^{\text{in-place}} &= U_{\text{sbp}}^{\text{in-place}} (-\Omega(x)),\\
U_{\text{poly}}^{\text{in-place}}(p) &= U_{\text{sbp}}^{\text{in-place}} (p(\Omega(x),\Omega(y),..)).\end{align}
We present an example calculation for the case if in-place additions on two registers in a computational basis state $\ket{x}\ket{y}$ with sizes $n, m$ are performed:
\begin{equation}
\begin{aligned}
&U_{\text{add}}^{\text{in-place}}\ket{x}\ket{y}\\
= &U_{\text{sbp}}^{\text{in-place}}(\Omega(x)) \ket{x}\ket{y}\\
= &\text{QFT}^{\dagger} U_{\mathrm{G}}(\Omega(x)) \text{QFT} \ket{x}\ket{y}\\
= &\frac{1}{\sqrt{2^{m}}} \text{QFT}^{\dagger} U_{\mathrm{G}}(\Omega(x)) \ket{x} \sum_{k = 0}^{2^{m}-1} \text{exp}\left(\frac{2\pi i y k}{2^{m}} \right)\ket{k}\\
= &\frac{1}{\sqrt{2^{m}}} \text{QFT}^{\dagger} \ket{x} \sum_{k = 0}^{2^{m}-1} \text{exp}\left(\frac{2\pi i (x + y) k}{2^{m}} \right)\ket{k}\\
= &\ket{x}\ket{(x + y)\text{ mod }2^{m}}
\end{aligned}
\end{equation}
\subsection{In-place Multiplication}
\label{in-place_mult}
It is well known that every quantum operation has to be invertible - this is because all the elementary gates that are available to us are invertible, implying that every sequence of them is also invertible. For an arbitrary modular in-place multiplication this is however unfortunately not the case. For instance, we have
\begin{align}
(2 \times 7) \text{ mod }8 = 6 \text{ mod }8\\
(2 \times 3) \text{ mod }8 = 6 \text{ mod }8
\end{align}
If we are only given the result 6 and the fact that our operation was a 2-multiplication, we can't infer which was the initial value. In fact, this is the reason why our modular arithmetic in $\mathbb{Z}/2^n \mathbb{Z}$ is a ring and not a field because not all of its elements possess an inverse. 
The following theorem addresses the invertibility of elements in $\mathbb{Z}/2^n \mathbb{Z}$.

\begin{theorem}
Let $n \in \mathbb{N}$ and $a \in \mathbb{Z}/2^n\mathbb{Z}$. We then have
\begin{align}
a \text{ invertible in } \mathbb{Z}/2^n \mathbb{Z} \Leftrightarrow a \text{ mod }{2} = 1
\end{align}
\end{theorem}
\begin{proof}
According to Bézout's lemma \cite{bezout1779}, the condition of $a$ being invertible is that
\begin{align}
\text{gcd}(a, 2^n) = 1,
\end{align}
where gcd denotes the greatest common divisor function. In our case this is easily evaluated because the only prime factor of $2^n$ is $2$, implying $a$ is invertible if $2$ is not a prime factor of $a$. This is equivalent to $2$ being an odd number.
\end{proof}
$\text{ }$ \newline
This implies that only circuits performing modular multiplications with odd numbers are reversible and can get executed in-place. Hence, the question emerges of how such circuits can be generated.\\ \\
For now, let's focus on multiplying with numbers that are classically known i.e. not the content of a quantum register. A possible way to do this within the framework of the in-place SB-polynomial encoder is to modify the initial QFT. Let $a,x \in \mathbb{Z}/2^n \mathbb{Z}$ where $a$ is invertible. The first step to derive this modification will now be to construct the Fourier-transformed $\ket{ax}$ state. By applying the inverse transform $\text{QFT}^{\dagger}$, $\ket{ax}$ can then afterward be recovered. According to Nielsen \& Chuang eq. (5.5-5.10) \cite{nielsen00} we have
\begin{equation}
\begin{aligned}
\text{QFT}\ket{ax} &= \frac{1}{\sqrt{2^n}} \bigotimes_{l = 0}^{n-1} \left( \ket{0} + \text{exp} \left(\frac{2\pi i ax}{2^{l+1}}\right) \ket{1}\right)\\
&= \frac{1}{\sqrt{2^n}} \bigotimes_{l = 0}^{n-1} \left( \ket{0} + \prod_{k = 0}^{n-1} \text{exp} \left(x_k\frac{2\pi i a 2^k}{2^{l+1}}\right) \ket{1}\right).
\end{aligned}
\end{equation}
Note that for index constellations where $k > l$, the phase is an integer multiple of $2\pi$. Hence we can simplify
\begin{equation}
\begin{aligned}
&\text{QFT}\ket{ax} \\
=&\frac{1}{\sqrt{2^n}} \bigotimes_{l = 0}^{n-1} \left( \ket{0} + \prod_{k = 0}^{l} \text{exp} \left(x_k\frac{2\pi i a 2^k}{2^{l+1}}\right) \ket{1}\right).
\end{aligned}
\end{equation}
Similar to the regular quantum Fourier-transform, this equation shows how the phases of the corresponding qubits are successively synthesized: The $\ket{1}$ state of the $l$-th qubit receives a phase if the $k$-th qubit is also in the $\ket{1}$ state because only then $x_k = 1$ (otherwise we have $x_k = 0$). This is done by performing controlled phase gates on these qubits. To turn this into an in-place multiplication, we will therefore simply perform a QFT whilst substituting any occurring $\text{CP}(\phi)$ gates\footnote{A CP gate is a controlled phase gate} with $\text{CP}(a\phi)$ gates. For an example circuit check fig. \ref{in_place_mult_cirquits}. There is however still one extra case left to be considered. If $l = k$, the gate which would apply the phase $x_k\pi$ in the regular QFT, is the H gate on the $k$-th qubit:
\begin{align}
\text{H} \ket{x_k} = \frac{1}{\sqrt{2}}(\ket{0} + \text{exp}(i\pi x_k)\ket{1}).
\end{align}
Since there is no such gate that would perform the same transformation but with a phase of $\pi a x_k $ this might cause problems. This is however easily resolved, because we assumed that $a$ is invertible, implying $a \text{ mod } 2 = 1$. Therefore
\begin{equation}
\begin{aligned}
(x_k a) \text{ mod } 2 &= (x_k \text{ mod } 2) (a \text{ mod } 2)\\
& = x_k \text{ mod } 2.
\end{aligned}
\end{equation}
Implying $x_k \pi$ and $a x_k \pi$ correspond to the same phase shift.\\
We conclude:
\begin{equation}
\begin{aligned}
&\text{QFT }U^{\text{in-place}}_{\text{mul}} (a)\\
= &\frac{1}{\sqrt{2^n}} \text{SW}\prod_{l = 0}^{n-1} \left(\prod_{k = 0}^{l-1} \text{CP}_{x_l x_k} \left(a\frac{2\pi 2^k }{2^{l+1}}\right)\text{H}_{x_l}\right),
\end{aligned}
\end{equation}
where SW denotes the swap gates, which are executed as usual. Note that even though this notation might suggest that $U^{\text{in-place}}_{\text{mul}} (a)$ could in principle be a standalone circuit, this is not the case - the in-place multiplication can only happen \textit{during} a $\text{QFT}$. We choose this notation to capture its functionality, which is
\begin{align}
\text{QFT }U^{\text{in-place}}_{\text{mul}} (a) \ket{x} = \text{QFT} \ket{ax \text{ mod } 2^n}.
\end{align}
For an example circuit, see fig. \ref{in_place_mult_cirquits}.

\subsection{Semi-in-place multiplication}
The restriction of only being able to perform odd number in-place multiplication can be at least partially lifted. For this, we note that a $2^k$ multiplication is the same as a $k$-bit shift into the more significant direction. This can be achieved via a compiler pass for increasing the exponent of the float in question by $k$. To perform in-place multiplication with an arbitrary (classically known) integer $a$, we, therefore, follow the following protocol
\begin{itemize}
\item Factorize $a$ (on the classical computer) such that $ a = b 2^k$, where b is odd.
\item Add $k$ to the exponent of the quantum float in question.
\item Perform in-place multiplication with $b$ as discussed in section \ref{in-place_mult}.
\end{itemize}
Even though this method in principle allows arbitrary number multiplication, it comes with some drawbacks
\begin{itemize}
\item It is still not possible to perform in-place multiplication with a number encoded in another quantum register.
\item The $2^k$ multiplication does not follow the modular arithmetic, while the $b$ multiplication does, which results in a complicated overflow behavior. 
\end{itemize}
To be more specific, consider the example $a  = 6 = 3 \cdot 2^1$ which will be in-place multiplied on the state unsigned integer state $\ket{[7]^0_4}$. Since the register size is $4$, the maximum number which can get represented is $2^4 -1 = 15$, implying the mantissa multiplication $3 \cdot 7 = 21$ results in overflow. Following the above protocol we have
\begin{align}
U_{\text{mul}}^{\text{semi-in-place}} (6) \ket{[7]^{0}_4} =  \ket{[3 \cdot 7]^{1}_4}
\end{align}
If we now measure and decode, we get the result 
\begin{align}
((3 \cdot 7) \text{ mod } 2^4) \cdot 2^{1} = 10,
\end{align}
which is rather uninformative. Note that this only causes problems, if the \textbf{mantissa multiplication} results in overflow: If we instead multiplied $a  = 6$ on $\ket{[3]^0_4}$, the correct result would be to be $18$, which is still more than $15$. However the mantissa multiplication $3 \cdot 3 = 9$ yields no overflow, so we still acquire the correct result
\begin{align}
((3 \cdot 3) \text{ mod } 2^4) \cdot 2^{1} = 18.
\end{align}

\newpage
\section{Performance}
\label{sec:performance}
This section will demonstrate some techniques that lower resource requirements of the SBP-encoder compared to the naive implementation. We note that even though these techniques can be used for floating point arithmetic evaluations, we formulate the equations in terms of unsigned integer arithmetic in order to suppress the variety of indices. Subsequently, a complexity analysis will be given.
\subsection{Conditional execution of $U_\mathrm{G}$}
\label{controlled_exec_U_G}
The general idea of the first improvement is to calculate the truth value of a given monomial into an ancilla qubit and use this ancilla qubit to control the execution of the $U_\mathrm{G}(y)$ gate.\\
This can be understood best with a hands on example. Imagine we want to encode the monomial $\Omega(x) = 3x_0 x_1 x_2$.
The naive implementation fig. \ref{naive_sbp_implementation} (as described in section \ref{sec:fundamentals}) would be to perform the phase gates in $U_\mathrm{G}(3)$ controlled on the qubits $x_0, x_1, x_2$ (see fig. \ref{naive_sbp_implementation}).
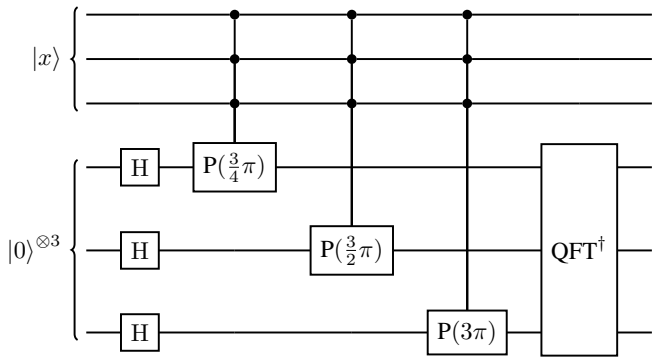
\begin{figure}
\adjustbox{max width=0.5\textwidth}{
\begin{quantikz}
\lstick[wires = 3]{$\ket{x}$} \qw &  \qw &\ctrl{3} &\ctrl{4} &\ctrl{5} &\qw & \qw \\
\qw & \qw &\ctrl{2}&\ctrl{3}&\ctrl{4} &\qw & \qw \\
\qw & \qw & \ctrl{1}&\ctrl{2} &\ctrl{3} &\qw & \qw \\
\lstick[wires = 3]{$\ket{0}^{\otimes 3}$} \qw & \gate{\mathrm{H}} & \gate{\text{P}(\frac{3}{4}\pi)} &\qw &\qw &\gate[3]{\text{QFT}^\dagger} & \qw \\
\qw & \gate{\mathrm{H}}&\qw & \gate{\text{P}(\frac{3}{2}\pi)} &\qw & & \qw\\
\qw & \gate{\mathrm{H}} &\qw &\qw & \gate{\text{P}(3\pi)} && \qw 
\end{quantikz}
}
\caption{\label{naive_sbp_implementation} Naive implementation of the SBP-encoder.}
\end{figure}\\
The improved version calculates the truth value of the product $x_0 x_1 x_2$ into an ancilla qubit, which is then used as a control for the phase gates (see fig. \ref{ancilla_supported_impl}).

\begin{figure}
\adjustbox{max width=0.5\textwidth}{
\begin{quantikz}
\lstick[wires = 3]{$\ket{x}$} 
\qw 	& \qw 		&\ctrl{6} 	&\qw 					& \qw 					&\qw 			& \ctrl{6}							& \qw								& \qw \\
\qw 	& \qw 		&\ctrl{5}	&\qw 					& \qw 					&\qw 			& \ctrl{5}							& \qw								& \qw \\
\qw 	& \qw 		& \ctrl{4}	&\qw 					& \qw 					&\qw 			& \ctrl{4}							& \qw 								& \qw\\
\lstick[wires = 3]{$\ket{0}^{\otimes 3}$}
\qw 	&\gate{\mathrm{H}} 	&\qw 		& \gate{\text{P}(\frac{3}{4}\pi)} & \qw 					& \qw			& \qw 								&\gate[3]{\text{QFT}^\dagger}		& \qw \\
\qw  	& \gate{\mathrm{H}}	&\qw 		&\qw 					& \gate{\text{P}(\frac{3}{2}\pi)} & \qw 			& \qw 								& 									& \qw \\
\qw  	& \gate{\mathrm{H}} 	&\qw 		&\qw 					& \qw 					& \gate{\text{P}(3\pi)} 	& \qw 								& 									& \qw \\
\lstick[wires = 1]{$\ket{0}^{\otimes 1}$}
\qw 	& \qw 		&\gate{\mathrm{X}} 	& \ctrl{-3} 			& \ctrl{-2} 			& \ctrl{-1}		& \gate{\mathrm{X}} 							& \qw 								& \qw
\end{quantikz}
}
\caption{\label{ancilla_supported_impl}Ancilla supported implementation of the SBP-encoder.}
\end{figure}
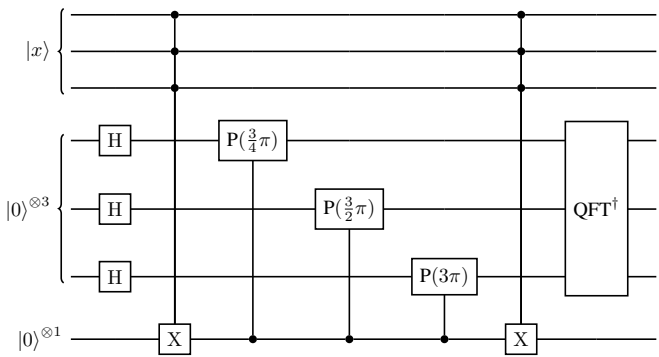
In this very simple scenario, the advantage is already visible as we used two instead of three 3-controlled operations. The improvements obviously rapidly grows, once we have a target register with more than three qubits.\\
The CNOT/RZ gate count for the (un)computation of the truth value of the ancilla qubit can furthermore be reduced by (asymptotically) $50\%$ by using phase tolerant gray synthesis, a synthesis technique we presented in \cite{Seidel2021}.
\subsection{SWAP-less Fourier 	transform}
The CNOT count can be reduced with omitting the SWAP gates at the end of the QFT implementation by reversing the qubit input order, i.e, reversing the order of the applied phase gates:
\begin{figure}

\adjustbox{max width=0.4\textwidth}{
\begin{quantikz}
\lstick[wires = 3]{$\ket{x}$} 
\qw 	& \qw 		&\ctrl{6} 	&\qw 					& \qw 					&\qw 			& \ctrl{6}							& \qw								& \qw \\
\qw 	& \qw 		&\ctrl{5}	&\qw 					& \qw 					&\qw 			& \ctrl{5}							& \qw								& \qw \\
\qw 	& \qw 		& \ctrl{4}	&\qw 					& \qw 					&\qw 			& \ctrl{4}							& \qw 								& \qw\\
\lstick[wires = 3]{$\ket{0}^{\otimes 3}$}
\qw 	&\gate{\mathrm{H}} 	&\qw 		& \gate{\text{P}(3\pi)} & \qw 					& \qw			& \qw 								&\gate[3]{\text{QFT}^\dagger \text{ (no swaps) }}		& \qw \\
\qw  	& \gate{\mathrm{H}}	&\qw 		&\qw 					& \gate{\text{P}(\frac{3}{2}\pi)} & \qw 			& \qw 								& 									& \qw \\
\qw  	& \gate{\mathrm{H}} 	&\qw 		&\qw 					& \qw 					& \gate{\text{P}(\frac{3}{4}\pi)} 	& \qw 								& 									& \qw \\
\lstick[wires = 1]{$\ket{0}^{\otimes 1}$}
\qw 	& \qw 		&\gate{\mathrm{X}} 	& \ctrl{-3} 			& \ctrl{-2} 			& \ctrl{-1}		& \gate{\mathrm{X}} 							& \qw 								& \qw
\end{quantikz}

}
\caption{\label{ancilla_supported_sbp} Ancilla supported implementation without SWAP gates at the end of the QFT.}
\end{figure}
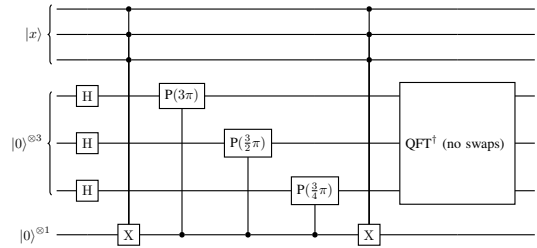
\subsection{Parallel $U_\mathrm{G}$ execution}
\label{parallel_U_G_execution}
Another technique that can be used to reduce the circuit depth is to allocate multiple ancilla qubits, calculate the truth value of various monomials into them and then execute the CP gates in parallel. An example for this is depicted in fig. \ref{depth_reduction_example}. For an unsigned $32$-bit integer multiplication, using $32$ ancilla qubits resulted in a depth reduction by a factor $15$ versus only a single available ancilla.

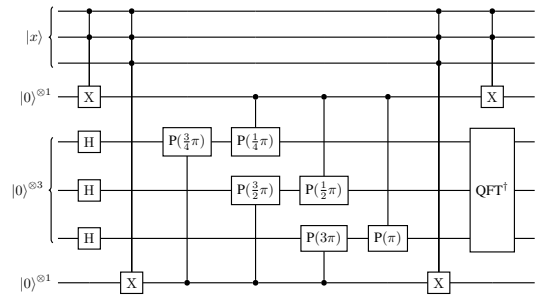
\begin{figure}
\adjustbox{max width=0.4\textwidth}{
\begin{quantikz}
\lstick[wires = 3]{$\ket{x}$} 
\qw 	& \ctrl{3} 	&\ctrl{7} 	&\qw 					& \qw 					&\qw 			&\qw 			& \ctrl{7}							& \ctrl{3}								& \qw \\
\qw 	& \ctrl{2} 	&\ctrl{6}	&\qw 					& \qw 					&\qw 			&\qw 			& \ctrl{6}							& \ctrl{2}								& \qw \\
\qw 	& \qw 		& \ctrl{5}	&\qw 					& \qw 					&\qw 			&\qw 			& \ctrl{5}							& \qw 								& \qw\\
\lstick[wires = 1]{$\ket{0}^{\otimes 1}$}
\qw 	& \gate{\mathrm{X}} 	&\qw		& \qw 					& \ctrl{1} 				& \ctrl{2} 		& \ctrl{3}							& \qw						& \gate{\mathrm{X}} 								& \qw\\
\lstick[wires = 3]{$\ket{0}^{\otimes 3}$}
\qw 	&\gate{\mathrm{H}} 	&\qw 		& \gate{\text{P}(\frac{3}{4}\pi)} & \gate{\text{P}(\frac{1}{4}\pi)}	& \qw			&\qw 			& \qw 								&\gate[3]{\text{QFT}^\dagger}		& \qw \\
\qw  	& \gate{\mathrm{H}}	&\qw 		&\qw 					& \gate{\text{P}(\frac{3}{2}\pi)} & \gate{\text{P}(\frac{1}{2}\pi)} 		&\qw 			& \qw 								& 									& \qw \\
\qw  	& \gate{\mathrm{H}} 	&\qw 		&\qw 					& \qw 					& \gate{\text{P}(3\pi)} 	&\gate{\text{P}(\pi)}		& \qw 								& 									& \qw \\
\lstick[wires = 1]{$\ket{0}^{\otimes 1}$}
\qw 	& \qw 		&\gate{\mathrm{X}} 	& \ctrl{-3} 			& \ctrl{-2} 			& \ctrl{-1}		&\qw 			& \gate{\mathrm{X}} 							& \qw 								& \qw
\end{quantikz}
}
\caption{\label{depth_reduction_example} Multi-ancilla supported parallel execution of two controlled $U_\mathrm{G}$ gates, encoding the SB-polynomial $\Omega(x) = 3x_0 x_1 x_2 + x_0 x_1$.}
\end{figure}
\subsection{Monomial encoding order}
Another important factor for circuit depth reduction is the order, in which monomials are encoded. An example for why this makes a difference is depicted in fig. \ref{monomial_order_figure}. To understand how much of an impact this can make, we note that the worst order we could find, has roughly about $17$ times the depth of the best order we could find for a 32-bit unsigned integer multiplication. Our general approach for determining a suited order, is to evaluate a cost function $\mathcal{C}(\Omega_\text{m})$ for each monomial $\Omega_\text{m}$ in each step and choose the monomial with the smallest cost for that particular iteration. We evaluated several cost functions and found that
\begin{align}
\mathcal{C}(\Omega_\text{m}) = \text{max}(\{ d_{x_i} | x_i \in \text{Vars}(\Omega_\text{m}) \})
\end{align}
produces the best results. Here, $\text{Vars}(\Omega_\text{m})$ denotes the set of variables of $\Omega_\text{m}$ (i.e., $\text{Vars}(4x_0 x_1 x_3) = \{x_0, x_1, x_3\}$) and $d_{x_i}$ the depth of the qubit\footnote{To define the depth of a qubit $q$, we divide the containing circuit $C$ (as with usual depth definition) in to a sequence of time steps. During each time step, each qubit can execute at most one gate. The depth of $q$ is then defined as the amount of time steps until the final gate on $q$ is executed. Therefore we have the following formula for the depth of the circuit: $D_C = \text{max}(\{d_q | q \text{ is qubit in } C\})$.} associated to $x_i$.
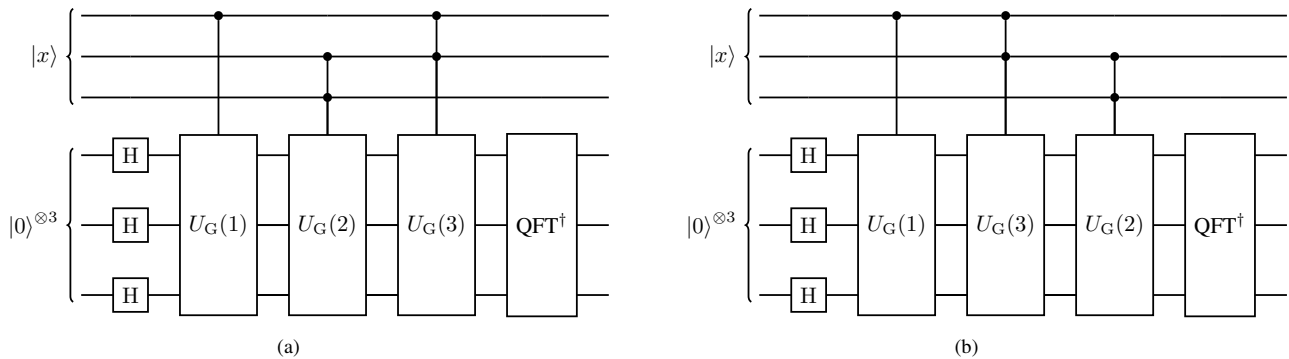
\begin{figure*}
\scalebox{0.85}{
\begin{subfigure}{0.5\textwidth}
\begin{quantikz}
\lstick[wires = 3]{$\ket{x}$} 
\qw 	& \qw 		&\ctrl{3}			&\qw 				&\ctrl{3} 			&\qw 							& \qw \\
\qw 	& \qw 		&\qw				&\ctrl{4}			&\ctrl{3}	 		&\qw 							& \qw \\
\qw 	& \qw 		&\qw				&\ctrl{3} 			&\qw 				&\qw 							& \qw \\
\lstick[wires = 3]{$\ket{0}^{\otimes 3}$} 
\qw 	& \gate{\mathrm{H}} 	& \gate[3]{U_\mathrm{G}(1)} 	& \gate[3]{U_\mathrm{G}(2)}	&\gate[3]{U_\mathrm{G}(3)}	&\gate[3]{\text{QFT}^\dagger} 	& \qw \\
\qw 	& \gate{\mathrm{H}}	&	 				& \qw				&\qw	 			& 								& \qw\\
\qw 	& \gate{\mathrm{H}} 	& 					& \qw 				& \qw		 		&								& \qw 
\end{quantikz}
\caption{\label{right_order}}
\end{subfigure}}
\hspace{1cm}
\scalebox{0.85}{
\begin{subfigure}{0.5\textwidth}
\begin{quantikz}
\lstick[wires = 3]{$\ket{x}$} 
\qw 	& \qw 		&\ctrl{3}			&\ctrl{4} 				&\qw 				&\qw 							& \qw \\
\qw 	& \qw 		&\qw				&\ctrl{3}				&\ctrl{4} 			&\qw 							& \qw \\
\qw 	& \qw 		&\qw				&\qw		 			&\ctrl{3} 			&\qw 							& \qw \\
\lstick[wires = 3]{$\ket{0}^{\otimes 3}$} 
\qw 	& \gate{\mathrm{H}} 	& \gate[3]{U_\mathrm{G}(1)} 	& \gate[3]{U_\mathrm{G}(3)}		&\gate[3]{U_\mathrm{G}(2)}	&\gate[3]{\text{QFT}^\dagger} 	& \qw \\
\qw 	& \gate{\mathrm{H}}	&	 				& \qw					&\qw 				& 								& \qw\\
\qw 	& \gate{\mathrm{H}} 	& 					& \qw 					& \qw		 		&								& \qw 
\end{quantikz}
\caption{\label{wrong_order}}
\end{subfigure}}
\caption{\label{monomial_order_figure} Two different orders of monomial encoding for the SB-polynomial $\Omega(x) = x_0  + 2 x_1 x_2 + 3 x_0 x_1$. We assume here that the techniques from section \ref{parallel_U_G_execution} are used to execute the controlled $U_\mathrm{G}$ gates in parallel. Note that \textbf{\ref{right_order}} allows faster execution because the ancillas for $U_\mathrm{G}(1)$ and $U_\mathrm{G}(2)$ can be computed in parallel. This is not the case for \textbf{\ref{wrong_order}}. Here, the computation of the corresponding ancillas awaits the conclusion of the previous controlled $U_\mathrm{G}$ gate in every step.}
\end{figure*}

\subsection{Entanglement via Global M\o lmer-S\o rensen gates}
One of the interesting perks of ion-trap architectures is that they are able to perform multiple two-qubit entangling operations within one laser pulse \cite{Moelmer1999}. These operations are frequently called \textit{Global M\o lmer-S\o rensen} or GMS gates. The unitary matrix describing a GMS gate on $n$ qubits is determined by a symmetric $n \times n$ matrix $\chi$ indicating XX interactions:
\begin{equation}
\begin{aligned}
\text{GMS}(\chi) = \text{exp}\left(-\frac{i}{2} \sum_{i = 0}^{n-1} \sum_{j = i+1}^{n-1} \chi_{ij} \text{X}_i \otimes \text{X}_j \right)
\end{aligned}
\end{equation}
An important distinction that has to be made is whether the entries of $\chi$ are uniform\footnote{Uniform denotes that $\forall i,j : \chi_{ij} = c$ for some constant $c$} or not. The uniform case has been demonstrated in \cite{Moelmer1999} - the non-uniform case has not been realized on a physical backend to the best of our knowledge, however we still include it's treatment in case of future progress on this field.\\
Using the techniques described in \cite{Maslov2018}, we can perform every single entangling operation using GMS gates. We will see that the requirement of uniform GMS gates only increases the amount of entangling operations by a constant factor (compared to allowing arbitrary $\chi$). This is possible because many parameters can be inserted into the circuit via parametrized single qubit gates.\\
For this, we note that the SBP-encoder consists of three types of entangling operations:
\begin{enumerate}
\item (Un)computing the ancilla qubit in fig. \ref{ancilla_supported_sbp}
\item Performing the $U_\mathrm{G}$ gate controlled on the ancilla qubit
\item Entangling operations contained in the QFT
\end{enumerate}
Regarding the first type, we refer to \cite{Maslov2018} which provides circuits for 2 and 3 controlled X gates. This obviously restricts the maximum degree of the SB-polynomials to 3, however we note that many important applications (such as multiplication, addition or subtraction) only require degree 2 or lower. Therefore, we leave the efficient synthesis of $\geq 4$ controlled X gates using only GMS entangling operations as an open research question.\\
For the second and third type we note that both are of the shape of multiple successive CP gates, where one knob is always on the same qubit but the others evenly distributed on the other qubits. We will denote circuits of this kind as \textit{ascending CP sequences} - for an example check fig. \ref{ascending_CP_ex}.
\begin{figure}
\begin{center}

\begin{quantikz}
\qw 	& \ctrl{1}						&\ctrl{2} 					&\ctrl{3}		 			&\ctrl{4}						& \qw\\
\qw 	& \gate{\text{P}(\phi_1)}	 	&\qw 						&\qw						&\qw							& \qw\\
\qw		& \qw 							& \gate{\text{P}(\phi_2)} 	&\qw						&\qw							& \qw\\
\qw 	& \qw 							& \qw						&\gate{\text{P}(\phi_3)}	&\qw 							& \qw\\
\qw 	& \qw  							& \qw						&\qw 						&\gate{\text{P}(\phi_4)}		& \qw\\
\end{quantikz}
\end{center}
\caption{\label{ascending_CP_ex} Apart from multi-controlled X gates, every entangling operation of the SBP-encoder is of this shape}
\end{figure}
We will now construct a technique that can execute ascending CP sequences, but also any other CP sequence with only a single non-uniform GMS gate (+some single qubit gates).\\
For a single CP gate acting on the two-qubit computational basis state $\ket{xy}$, where $x,y \in \{0,1\}$, the applied phase reads
\begin{equation}
\begin{aligned}
&\text{CP}_{ab}(\phi) \ket{xy} = \text{exp}(i\phi xy) \ket{xy}\\
&= \text{exp} \left( i\frac{\phi}{2}(x + y - (x \oplus y))\right)  \ket{xy}\\
&= \text{exp} \left(\frac{i\phi x}{2}\right) \text{exp}\left(\frac{i \phi y}{2}\right) \text{exp}\left( i\frac{-\phi (1 - (-1)^{x+y})}{4}\right) \ket{xy}\\
&= \text{P}_a\left(\frac{\phi}{2}\right) \text{P}_b\left(\frac{\phi}{2}\right) \text{exp}\left( i\frac{\phi}{4}\text{Z}_a \otimes \text{Z}_b\right) \text{exp}\left( i\frac{\phi}{4}\right) \ket{xy},
\end{aligned}
\label{CP_decomposition}
\end{equation}
where we used 
\begin{align}
2xy = x + y - (x \oplus y)\text{ } , \forall (x,y) \in \{0,1\}^2
\end{align} 
with the $\text{mod } 2$ addition denoted by $\oplus$. Since we proved in eq. \ref{CP_decomposition} for a complete base of $(\mathbb{C}^2)^{\otimes 2} $, we have a valid operator identity:
\begin{equation}
\label{CP_dissolution_eq}
\begin{aligned}
& \text{CP}(\phi)_{ab}
\\= \text{ } &\text{P}_a\left(\frac{\phi}{2}\right) \text{P}_b\left(\frac{\phi}{2}\right) \text{exp}\left( i\frac{\phi}{4}\text{Z}_a \otimes \text{Z}_b\right) \text{exp}\left( i\frac{\phi}{4}\right)
\end{aligned}
\end{equation}
\\
Furthermore, since every operator in \ref{CP_dissolution_eq} commutes, we conclude that any sequence of CP gates $(\text{CP}_{a_0 b_0}(\phi_0), .. ,\text{CP}_{a_m b_m}(\phi_m))$ on $n$ qubits can be represented by the unitary
\begin{equation}
\begin{aligned}
&\prod_{i = 0}^{m} \text{CP}_{a_i b_i}(\phi_i)\\
=&\text{exp}(i\Omega)\left(\prod_{i = 0}^{n-1}\text{P}_i(\omega_i)\right)\\
& \cdot \text{ exp}\left(-\frac{i}{2} \sum_{i = 0}^{n-1} \sum_{j = i+1}^{n-1} \chi_{ij} \text{Z}_i \otimes \text{Z}_j \right)\\
=&\text{exp}(i\Omega)\left(\prod_{i = 0}^{n-1}\text{P}_i(\omega_i)\right)\\
\cdot &\text{ exp}\left(-\frac{i}{2} \sum_{i = 0}^{n-1} \sum_{j = i+1}^{n-1} \chi_{ij} \mathrm{H}^{\otimes n} (\text{X}_i \otimes \text{X}_j) \mathrm{H}^{\otimes n} \right)\\
= &\text{exp}(i\Omega)\left(\prod_{i = 0}^{n-1}\text{P}_i(\omega_i)\right)\mathrm{H}^{\otimes n}\text{GMS}(\chi) \mathrm{H}^{\otimes n},
\end{aligned}
\end{equation}
where we used that for any linear operator $A$
\begin{equation}
\begin{aligned}
\mathrm{H}^{\otimes n} \text{exp}(i A) \mathrm{H}^{\otimes n} &= \mathrm{H}^{\otimes n} \sum_{k = 0}^{\infty}\frac{(iA)^k}{k!} \mathrm{H}^{\otimes n}\\
& = \sum_{k = 0}^{\infty}\frac{(i\mathrm{H}^{\otimes n}A \mathrm{H}^{\otimes n})^k}{k!}\\
&= \text{exp}\left( i\mathrm{H}^{\otimes n}A\mathrm{H}^{\otimes n}\right)
\end{aligned}
\end{equation}
The formulas for the parameters are
\begin{equation}
\begin{aligned}
\omega_i &= \sum_{\substack{k \leq m \\a_k = i \text{ or } b_k = i}} \frac{\phi_k}{2},\\
\chi_{ij} &= -\sum_{\substack{k \leq m \\a_k = i \text{ and } b_k = j}} \frac{\phi_k}{2},\\
\Omega &= \sum_{k \leq m} \frac{\phi_k}{4}.
\end{aligned}
\end{equation}
In case the hardware only supports uniform GMS gates, arbitrary CP sequences are no longer easily encodable, however, ascending CP sequences still are. For this, we write out the decomposition of a CP gate into CNOT gates:\\ \\

\begin{quantikz}
& \ctrl{1}		& \qw \\
& \gate{\text{P}(\phi)}& \qw
\end{quantikz}
=\begin{quantikz}
& \ctrl{1} & \gate{\text{P}\left(\frac{\phi}{2}\right)} & \ctrl{1} &\qw \\
& \targ{} & \gate{\text{P}\left(\frac{-\phi}{2}\right)} & \targ{} & \gate{\text{P}\left(\frac{\phi}{2}\right)}
\end{quantikz}\\ \\
Applying this to an ascending CP sequence yields\\
\begin{center}
\begin{quantikz}
\qw 	& \ctrl{1}						&\ctrl{2} 					&\ctrl{3}		 			&\ctrl{4}						& \qw\\
\qw 	& \gate{\text{P}(\phi_1)}	 	&\qw 						&\qw						&\qw							& \qw\\
\qw		& \qw 							& \gate{\text{P}(\phi_2)} 	&\qw						&\qw							& \qw\\
\qw 	& \qw 							& \qw						&\gate{\text{P}(\phi_3)}	&\qw 							& \qw\\
\qw 	& \qw  							& \qw						&\qw 						&\gate{\text{P}(\phi_4)}		& \qw\\
\end{quantikz}\\
=\\\vspace{0.5cm}
\begin{quantikz}
&\ctrl{4} 	& \gate{\text{P}\left(\sum_i \frac{\phi_i}{2}\right)}	& \ctrl{4} 	&\qw 											&\qw \\
&\targ{} 	& \gate{\text{P}\left(\frac{-\phi_1}{2}\right)}		 	& \targ{} 	&\gate{\text{P}\left(\frac{\phi_1}{2}\right)}	&\qw \\
&\targ{}	& \gate{\text{P}\left(\frac{-\phi_2}{2}\right)} 		& \targ{} 	&\gate{\text{P}\left(\frac{\phi_1}{2}\right)}	&\qw \\
&\targ{} 	& \gate{\text{P}\left(\frac{-\phi_3}{2}\right)} 		& \targ{}	&\gate{\text{P}\left(\frac{\phi_1}{2}\right)}	&\qw \\
&\targ{} 	& \gate{\text{P}\left(\frac{-\phi_4}{2}\right)}  		& \targ{}	&\gate{\text{P}\left(\frac{\phi_1}{2}\right)}	&\qw
\end{quantikz}
\end{center}
Such sequences of CNOT gates, with a single control and multiple targets are called \textit{fan-out} gates and can be realized with two uniform GMS gates according to \cite{Zeng2005}. Since we need two fan-out gates, we end up with 4 GMS gates per ascending CP sequence.
\subsection{Complexity analysis}

To estimate the performance of the processing of arithmetic operands of magnitude $N \in \mathbb{N}$, we separate the SBP-encoder into three steps
\begin{enumerate}
\item Initial H gates / Fourier transform (in the case of in-place operations)
\item (Controlled) $U_\mathrm{G}$-application
\item Fourier transform
\end{enumerate}
The initial H gates can be executed in constant time. The Fourier transform can be done in $\mathcal{O}(\text{log}(N)^2)$. If the hardware provides access to GMS gates, this reduces to $\mathcal{O}(\text{log}(N))$, as there is one ascending CP sequence per qubit.\\
In order to estimate the complexity of a single controlled $U_\mathrm{G}$ gate, we note that (un-)computing the truth value of the ancilla qubit (see section \ref{controlled_exec_U_G}) can be done in constant time. The next step is the application of an ascending CP sequence which has complexity $\mathcal{O}(\text{log}(N))$ for a CNOT based implementation and $\mathcal{O}(1)$ for GMS based implementations. We now estimate how many controlled $U_\mathrm{G}$ we have to apply. In the case of an addition the SB-polynomial is
\begin{align}
\Omega_{\text{add}}(x,y) = \sum_{i = 0}^{n_1-1} 2^i x_i + \sum_{j = 0}^{n_2-1} 2^j y_j,
\end{align}
which implies $\mathcal{O}(\text{log}(N))$ monomials. For multiplications we have
\begin{align}
\Omega_{\text{mult}}(x,y) = \left(\sum_{i = 0}^{n_1-1} 2^i x_i\right) \times \left( \sum_{j = 0}^{n_2-1} 2^j y_j \right).
\end{align}
This yields $\mathcal{O}(\text{log}(N)^2)$ monomials, as we have to iterate over every combination of $(i,j) \ i<n_1, j<n_2$. From this, we conclude that in the case of addtion, the controlled $U_\mathrm{G}$-application step can be performed in $\mathcal{O}(\text{log}(N)^2)$ (GMS: $\mathcal{O}(\text{log}(N))$) and in the case of multiplication we have $\mathcal{O}(\text{log}(N)^3)$ (GMS: $\mathcal{O}(\text{log}(N)^2)$).\\
The final Fourier transform has the same complexity as the initial thus leaving the complexity valid without modification. We summarize the results of our estimations in table \ref{complexity_estimation}.
\begin{table}[h]
\begin{tabular}{l | c c}
Operation & CNOT complexity & GMS complexity\\
\hline
Addition & $\mathcal{O}(\text{log}(N)^2)$ & $\mathcal{O}(\text{log}(N))$\\
Multiplication & $\mathcal{O}(\text{log}(N)^3)$ & $\mathcal{O}(\text{log}(N)^2)$\\
Deg. $k$ Polynomial & $\mathcal{O}(\text{log}(N)^{k+1})$ & $\mathcal{O}(\text{log}(N)^k)$\\
\end{tabular}
\caption{ \label{complexity_estimation} Gate-count complexity estimation of different arithmetic operations using either CNOT or global M\o lmer-S\o rensen gates as elementary entangling gates}
\end{table}
\vskip -3cm
\section{Critical points}
\label{sec:critical_points}
A valid point of criticism is the fact that ripple-carry adder based implementations of arithmetic evaluations scale as $\mathcal{O}(\text{log}(N))$ for additions and $\mathcal{O}(\text{log}(N)^2)$ for multiplications (check section \ref{sec:overview}) and therewith better compared to our approach. Apart from the possibility to reduce the complexity using GMS gates, this can be addressed by noting that the determined complexities are only valid regarding the gate count and not the depth. As it turns out (see fig \ref{depth_comparison_plot}) the depth (and with it the execution speed) is significantly lower than ripple-carry adder based implementations, since it can be reduced by one factor of $\text{log}(N)$ for each operation type.\\
To see how this works, we assume that the target register has a size of the order of $n_{\text{targ}} \approx \text{log}(N)$. We then allocate $n_{\text{targ}}$ ancilla qubits and apply the technique described in section \ref{parallel_U_G_execution}. This allows parallel execution of $n_{\text{targ}} \approx \text{log}(N)$ controlled $U_\mathrm{G}$ gates, implying an improvement in execution speed by this factor.\\
Another point of criticism is constituted by the fact that with growing factor size, the RZ gate phases in the $U_\mathrm{G}$ and QFT gates become exponentially small. Assuming only finite physical RZ gate precision this implies an upper limit for the operand size. For multiplications, the speed of our methods can however still be harnessed by reducing the problem into smaller multiplications using Karatsuba's divide-and-conquer algorithm \cite{Karatsuba1962}. This however comes at the cost of additional ancilla qubits for storing intermediate results.
\begin{figure*}
\begin{center}
\begin{subfigure}{.42\textwidth}
\includegraphics[width = \textwidth]{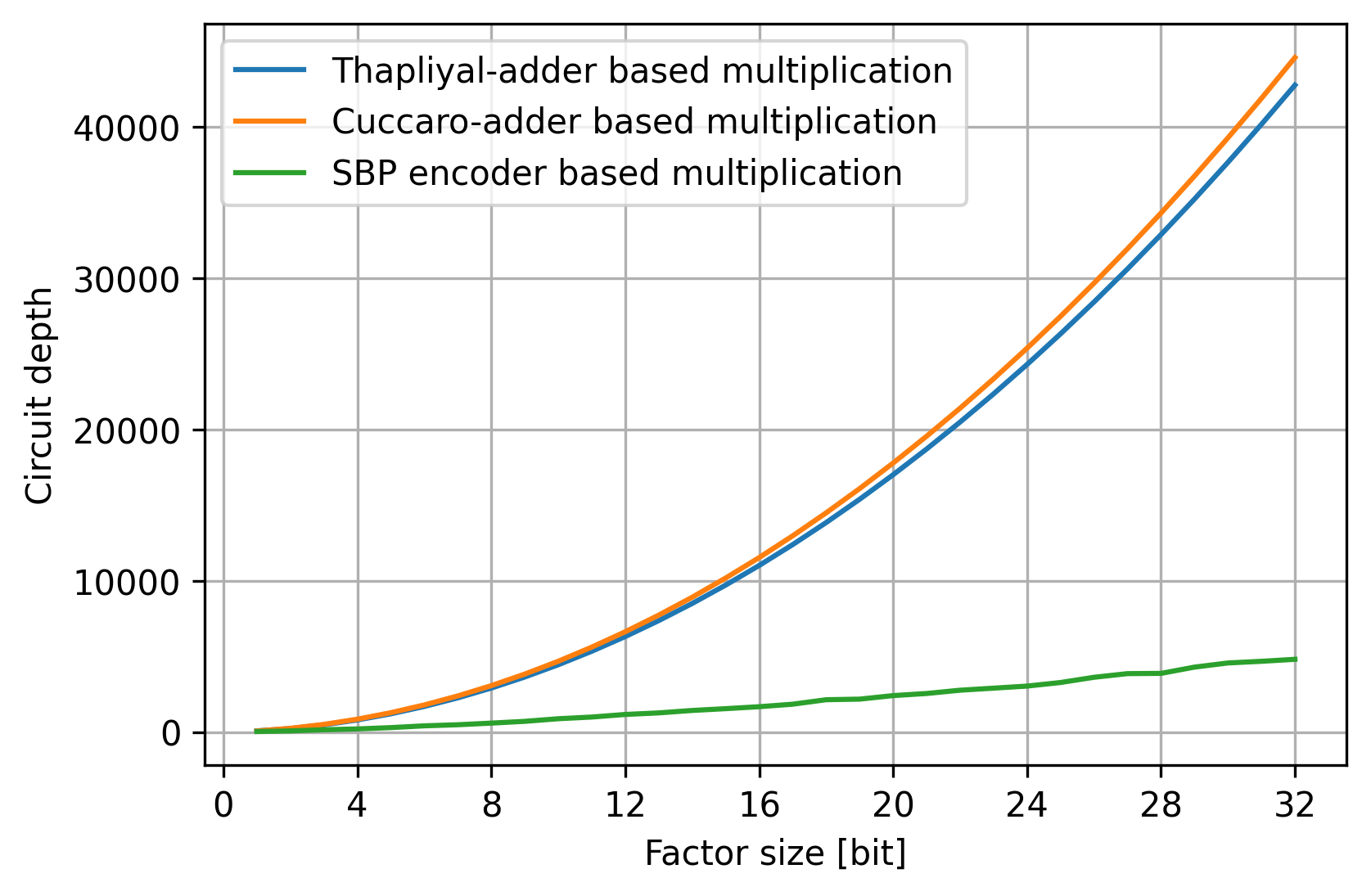}
\caption{\label{depth_comparison_plot}}
\end{subfigure}
\begin{subfigure}{.42\textwidth}
\includegraphics[width = \textwidth]{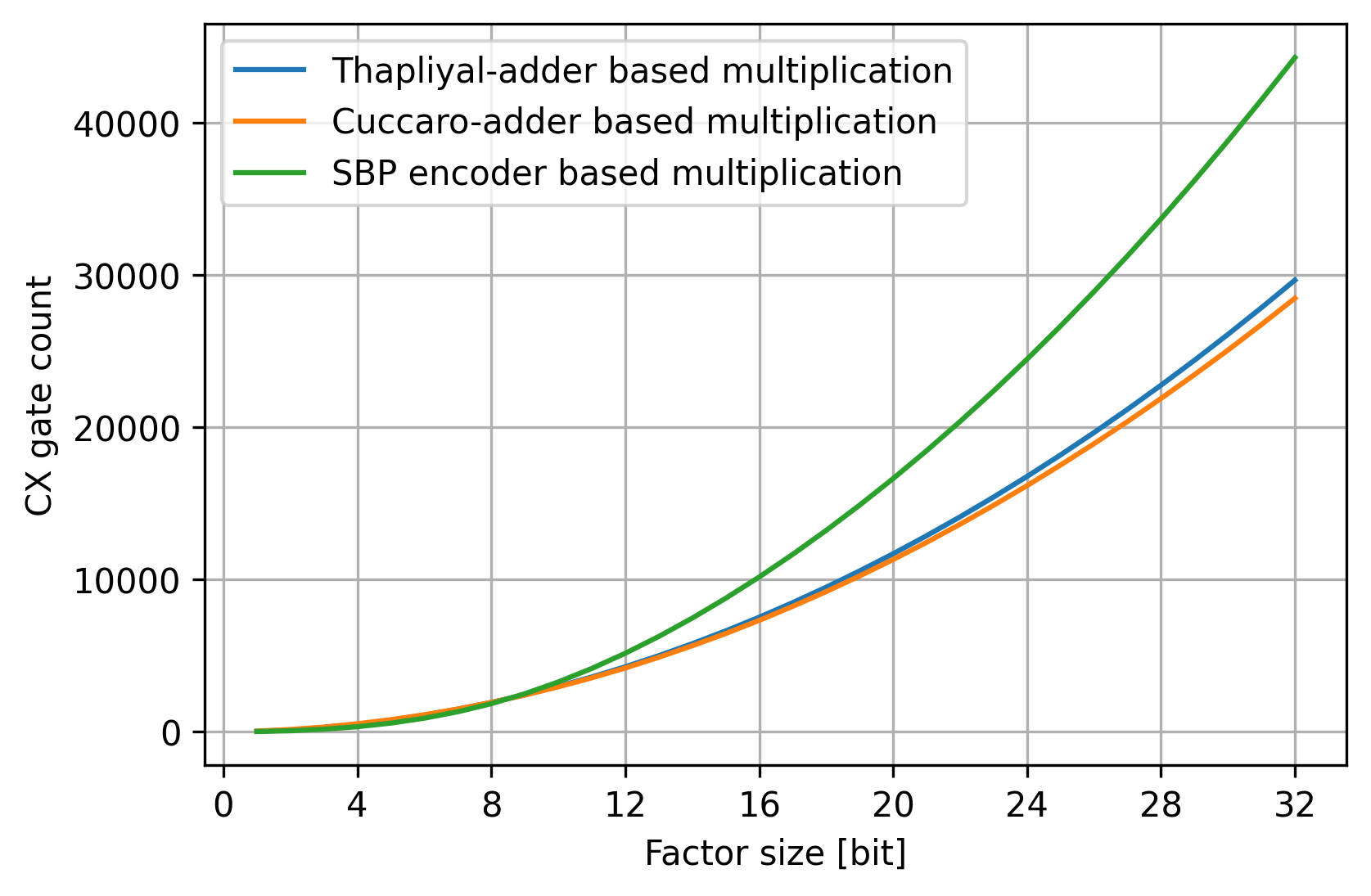}
\caption{\label{cx_gate_count_comparison_plot}}
\end{subfigure}
\begin{subfigure}{.42\textwidth}
\includegraphics[width = \textwidth]{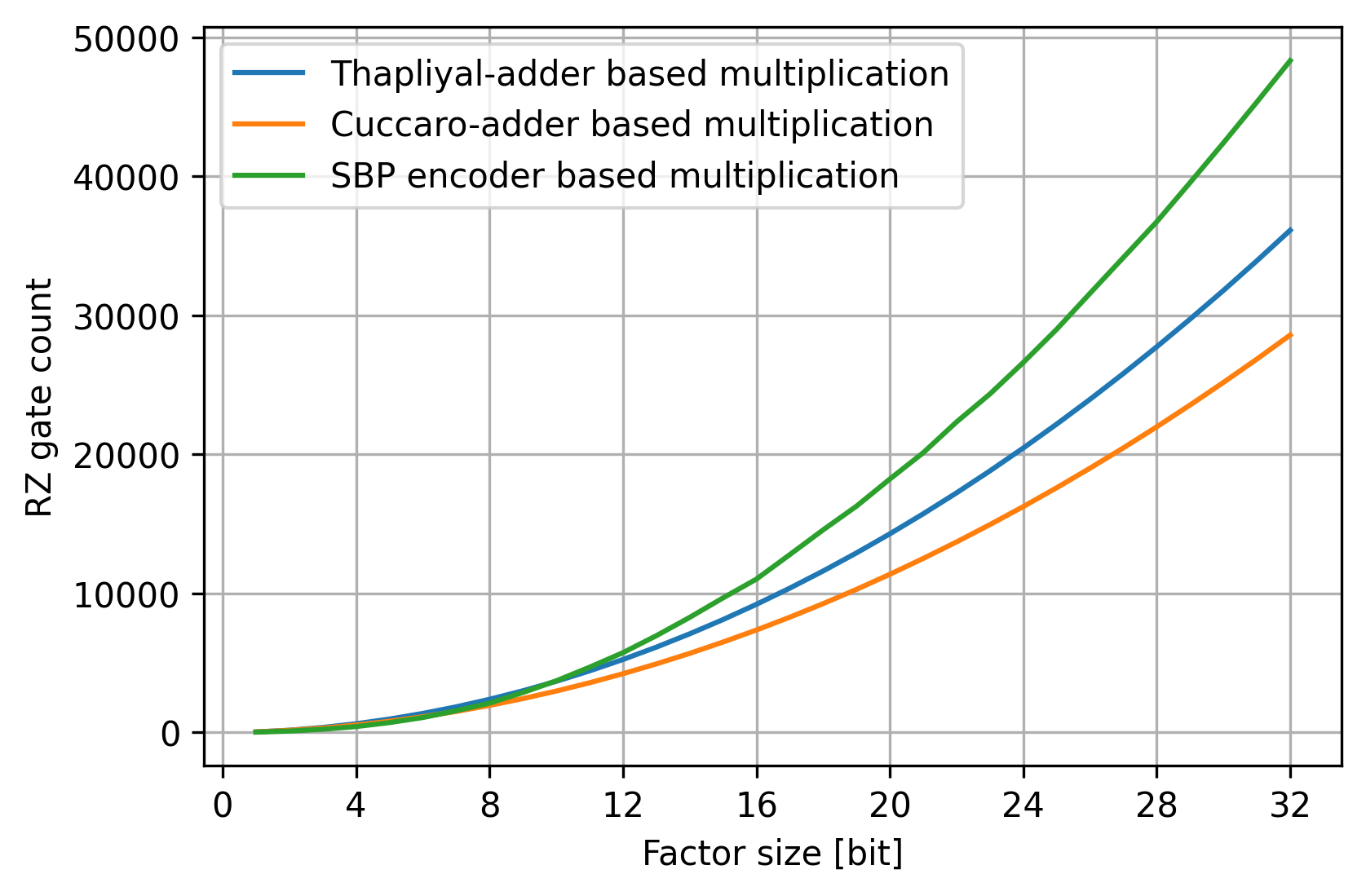}
\caption{\label{rz_gate_count_comparison_plot}}
\end{subfigure}
\begin{subfigure}{.42\textwidth}
\includegraphics[width = \textwidth]{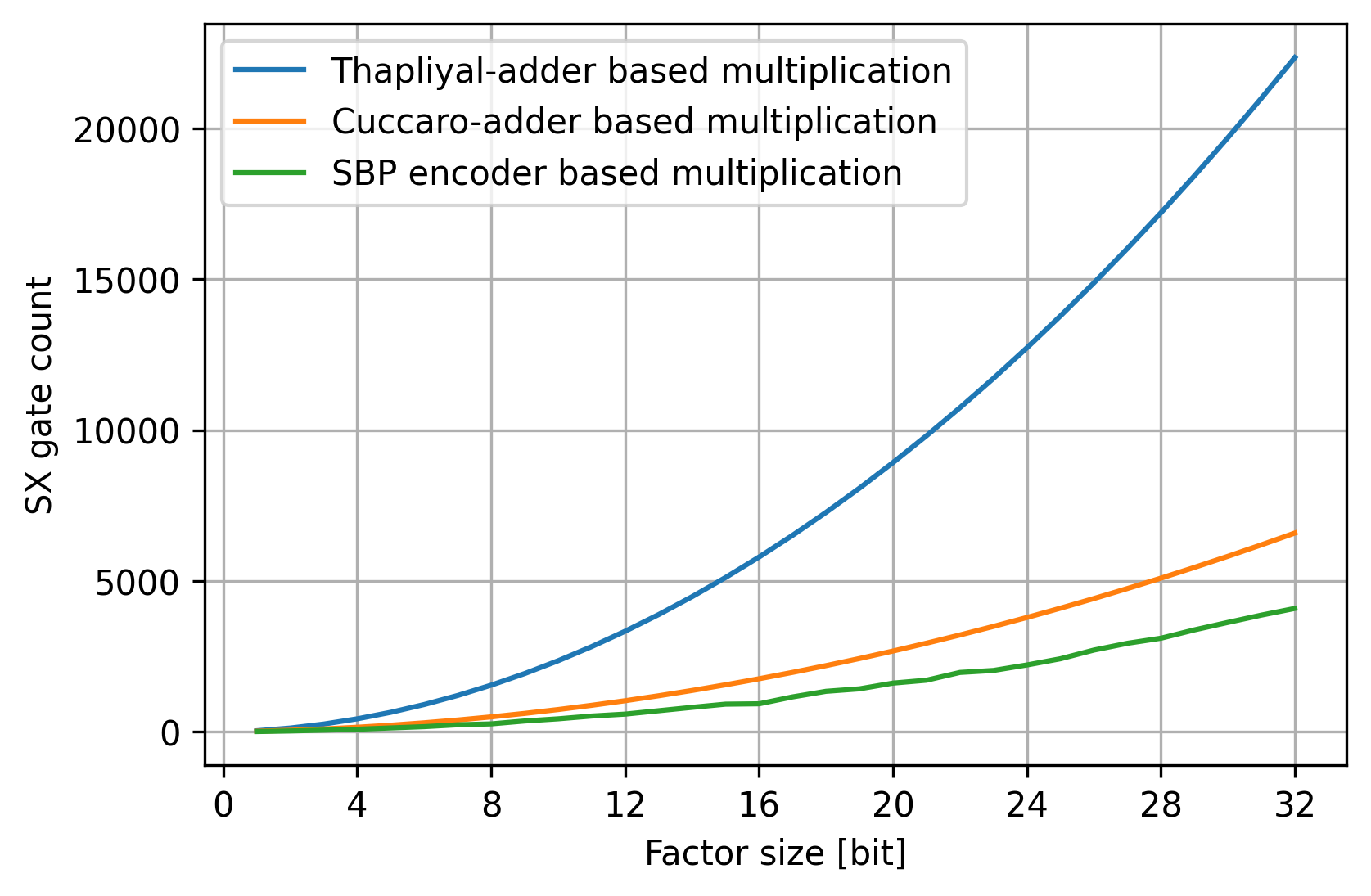}
\caption{\label{sx_gate_count_comparison_plot}}
\end{subfigure}
\end{center}
\caption{\label{performances_plots} \textbf{\ref{depth_comparison_plot}} Depth of the resulting circuits from an unsigned integer multiplication of our method compared to two other ripple-carry adder based multiplication approaches (\cite{Cuccaro2004}, \cite{Thapliyal2017}). For implementation specifics check section \ref{sec:overview}. At 32 bit, the resulting depth of the SBP encoder is only $10.7\%$ of the depth of the ripple-carry based approaches. Note that we chose the target register size as minimal as possible without risking overflow. The depths are calculated after transpiling the circuits into the gate set $\{ \text{CX}, \text{RZ}, \text{SX}\}$. The circuit construction, transpilation (optimization level 2) and evaluation was performed using IBM's \textit{Qiskit} \cite{Qiskit}. \textbf{\ref{cx_gate_count_comparison_plot}},\textbf{\ref{rz_gate_count_comparison_plot}},\textbf{\ref{sx_gate_count_comparison_plot}} Plot of the comparison of the gate counts of the previously mentioned approaches.}
\end{figure*}

\section{Conclusions}
\label{sec:conclusions}
In this work, we presented the design and handling of various types of arithmetic operations. For unsigned integer arithmetic, we constructed our circuits based on the idea that such evaluations can be written as SB-polynomials. Our way of encoding signed integers allows for an efficient operation evaluation since it can be reduced to the \textbf{un}signed integer arithmetic circuits. We furthermore extended our algorithm/encoding by the possibility to inter-operate between registers of different sizes. Regarding the representation of non-integers, it turns out that encoding both the mantissa and the exponent as quantum variables (\textit{bi-quantum encoding}) is in principle possible but only at the cost of highly increased circuit complexity. Instead, we store the exponent as a classical value (\textit{mono-quantum encoding}) which again enables us to profit from the efficient unsigned integer circuits.  Values encoded by our method can also be processed using other types of algorithms which perform modular arithmetic since there is no other specific requirement than the modular overflow behavior.\\
Subsequently, we demonstrate methods to perform in-place operations, which can both save on qubits and prevent the necessity of uncomputation operations.\\
Finally, we discuss several implementation improvements for the SB-polynomial encoder. Probably most outstanding is the possibility to perform every entangling operation using ion-trap native GMS gates. These gates allow the entanglement of more than two qubits within a single pulse, which enables a reduction in entangling gate count by a factor of $\mathcal{O}(\text{log}(N))$. We estimate our circuit complexity and find that CNOT based implementations of addition and multiplication circuits are asymptotically more expensive in gate count by one factor of $\mathcal{O}(\text{log}(N))$ compared to ripple-carry based approaches. This is however overshadowed by the fact that encoding SB-polynomials can be performed with many parallel gate executions. Using this we could show that regarding depth (and with that speed), our approach has no extra $\mathcal{O}(\text{log}(N))$ factor and furthermore provides a speed-up with a factor of more than 900\% (compared to carry ripple approaches).

\bibliographystyle{IEEEtran}
\bibliography{sources}

\begin{thebibliography}{10}
\providecommand{\url}[1]{#1}
\csname url@samestyle\endcsname
\providecommand{\newblock}{\relax}
\providecommand{\bibinfo}[2]{#2}
\providecommand{\BIBentrySTDinterwordspacing}{\spaceskip=0pt\relax}
\providecommand{\BIBentryALTinterwordstretchfactor}{4}
\providecommand{\BIBentryALTinterwordspacing}{\spaceskip=\fontdimen2\font plus
\BIBentryALTinterwordstretchfactor\fontdimen3\font minus
  \fontdimen4\font\relax}
\providecommand{\BIBforeignlanguage}[2]{{%
\expandafter\ifx\csname l@#1\endcsname\relax
\typeout{** WARNING: IEEEtran.bst: No hyphenation pattern has been}%
\typeout{** loaded for the language `#1'. Using the pattern for}%
\typeout{** the default language instead.}%
\else
\language=\csname l@#1\endcsname
\fi
#2}}
\providecommand{\BIBdecl}{\relax}
\BIBdecl

\bibitem{Cerezo2021}
\BIBentryALTinterwordspacing
M.~Cerezo, A.~Arrasmith, R.~Babbush, S.~C. Benjamin, S.~Endo, K.~Fujii, J.~R.
  McClean, K.~Mitarai, X.~Yuan, L.~Cincio, and et~al., ``Variational quantum
  algorithms,'' \emph{Nature Reviews Physics}, vol.~3, no.~9, p. 625–644, Aug
  2021. [Online]. Available: \url{http://dx.doi.org/10.1038/s42254-021-00348-9}
\BIBentrySTDinterwordspacing

\bibitem{Schaefer2018}
\BIBentryALTinterwordspacing
V.~M. Schäfer, C.~J. Ballance, K.~Thirumalai, L.~J. Stephenson, T.~G.
  Ballance, A.~M. Steane, and D.~M. Lucas, ``Fast quantum logic gates with
  trapped-ion qubits,'' \emph{Nature}, vol. 555, no. 7694, p. 75–78, Mar
  2018. [Online]. Available: \url{http://dx.doi.org/10.1038/nature25737}
\BIBentrySTDinterwordspacing

\bibitem{Cuccaro2004}
\BIBentryALTinterwordspacing
S.~A.~K. Steven A.~Cuccaro, Thomas G.~Draper and D.~P. Moulton, ``A new quantum
  ripple-carry addition circuit,'' Oct 2004. [Online]. Available:
  \url{https://arxiv.org/abs/quant-ph/0410184}
\BIBentrySTDinterwordspacing

\bibitem{shende2006}
\BIBentryALTinterwordspacing
V.~Shende, S.~Bullock, and I.~Markov, ``Synthesis of quantum-logic circuits,''
  \emph{IEEE Transactions on Computer-Aided Design of Integrated Circuits and
  Systems}, vol.~25, no.~6, p. 1000–1010, Jun 2006. [Online]. Available:
  \url{http://dx.doi.org/10.1109/TCAD.2005.855930}
\BIBentrySTDinterwordspacing

\bibitem{Thapliyal2021}
H.~Thapliyal, E.~MuÑoz-Coreas, T.~S.~S. Varun, and T.~S. Humble, ``Quantum
  circuit designs of integer division optimizing t-count and t-depth,''
  \emph{IEEE Transactions on Emerging Topics in Computing}, vol.~9, no.~2, pp.
  1045--1056, 2021.

\bibitem{Gilliam2021groveradaptive}
\BIBentryALTinterwordspacing
A.~Gilliam, S.~Woerner, and C.~Gonciulea, ``Grover {A}daptive {S}earch for
  {C}onstrained {P}olynomial {B}inary {O}ptimization,'' \emph{{Quantum}},
  vol.~5, p. 428, Apr. 2021. [Online]. Available:
  \url{https://doi.org/10.22331/q-2021-04-08-428}
\BIBentrySTDinterwordspacing

\bibitem{nielsen00}
M.~A. Nielsen and I.~L. Chuang, \emph{Quantum Computation and Quantum
  Information}.

\bibitem{Draper2000}
\BIBentryALTinterwordspacing
T.~G. Draper, ``Addition on a quantum computer,'' Aug 2000. [Online].
  Available: \url{https://arxiv.org/abs/quant-ph/0008033}
\BIBentrySTDinterwordspacing

\bibitem{Ruiz-Perez2017}
\BIBentryALTinterwordspacing
L.~Ruiz-Perez and J.~C. Garcia-Escartin, ``Quantum arithmetic with the quantum
  fourier transform,'' \emph{Quantum Information Processing}, vol.~16, no.~6,
  Apr 2017. [Online]. Available:
  \url{http://dx.doi.org/10.1007/s11128-017-1603-1}
\BIBentrySTDinterwordspacing

\bibitem{Oliveira2007}
D.~Oliveira and R.~Ramos, ``Quantum bit string comparator: Circuits and
  applications,'' \emph{Quantum Computers and Computing}, vol.~7, 01 2007.

\bibitem{bezout1779}
E.~BEZOUT, \emph{\BIBforeignlanguage{fre}{Théorie générale des équations
  algébriques}}.\hskip 1em plus 0.5em minus 0.4em\relax Paris: De l'Imprimerie
  de Ph.-D. Pierres, 1779.

\bibitem{Seidel2021}
\BIBentryALTinterwordspacing
R.~Seidel, C.~K.-U. Becker, S.~Bock, N.~Tcholtchev, I.-D. Gheorge-Pop, and
  M.~Hauswirth, ``Automatic generation of grover quantum oracles for arbitrary
  data structures,'' 2021. [Online]. Available:
  \url{https://arxiv.org/abs/2110.07545}
\BIBentrySTDinterwordspacing

\bibitem{Moelmer1999}
\BIBentryALTinterwordspacing
K.~Mølmer and A.~Sørensen, ``Multiparticle entanglement of hot trapped
  ions,'' \emph{Physical Review Letters}, vol.~82, no.~9, p. 1835–1838, Mar
  1999. [Online]. Available:
  \url{http://dx.doi.org/10.1103/PhysRevLett.82.1835}
\BIBentrySTDinterwordspacing

\bibitem{Maslov2018}
\BIBentryALTinterwordspacing
D.~Maslov and Y.~Nam, ``Use of global interactions in efficient quantum circuit
  constructions,'' vol.~20, no.~3, p. 033018, mar 2018. [Online]. Available:
  \url{https://doi.org/10.1088/1367-2630/aaa398}
\BIBentrySTDinterwordspacing

\bibitem{Zeng2005}
\BIBentryALTinterwordspacing
B.~Zeng, D.~L. Zhou, and L.~You, ``Measuring the parity of ann-qubit state,''
  \emph{Physical Review Letters}, vol.~95, no.~11, Sep 2005. [Online].
  Available: \url{http://dx.doi.org/10.1103/PhysRevLett.95.110502}
\BIBentrySTDinterwordspacing

\bibitem{Karatsuba1962}
Y.~O. A.~Karatsuba, ``Multiplication of many-digital numbers by automatic
  computers,'' \emph{Dokl. Akad. Nauk SSSR}, vol. 145, Feb 1962.

\bibitem{Thapliyal2017}
\BIBentryALTinterwordspacing
H.~Thapliyal and N.~Ranganathan, ``Design of efficient reversible logic-based
  binary and bcd adder circuits,'' \emph{ACM Journal on Emerging Technologies
  in Computing Systems}, vol.~9, no.~3, p. 1–31, Sep 2013. [Online].
  Available: \url{http://dx.doi.org/10.1145/2491682}
\BIBentrySTDinterwordspacing

\bibitem{Qiskit}
\BIBentryALTinterwordspacing
G.~Aleksandrowicz and T.~A. et~al., ``Qiskit: An open-source framework for
  quantum computing,'' jan 2019. [Online]. Available:
  \url{https://doi.org/10.5281/zenodo.2562111}
\BIBentrySTDinterwordspacing

\end{thebibliography}

\end{document}